\documentclass[superscriptaddress,prl,reprint]{revtex4-2}
\usepackage[english]{babel}
\usepackage[hidelinks]{hyperref}
\usepackage{amsmath}
\usepackage{amssymb}
\usepackage{amsthm}
\usepackage{graphicx}
\usepackage{natbib}
\usepackage{xcolor}
\usepackage{algorithm} 
\usepackage{algpseudocode}
\usepackage{etoolbox}

\newtheorem{theorem}{Theorem}
\newtheorem{proposition}{Proposition}
\theoremstyle{definition}
\newtheorem{definition}{Definition}
\newcommand{\norm}[1]{\left\lVert#1\right\rVert}

\begin{document}
\selectlanguage{english}

\title{Finding spectral gaps in quasicrystals}
\author{Paul Hege}
\affiliation{Mathematisches Institut, University of Tübingen, Auf der Morgenstelle 10, 72076 Tübingen, Germany}
\affiliation{Department of Neural Dynamics and MEG, Hertie Institute for Clinical Brain Research, University of Tübingen, Otfried-Müller-Str.~25, 72076 Tübingen, Germany}
\author{Massimo Moscolari}
\affiliation{Mathematisches Institut, University of Tübingen, Auf der Morgenstelle 10, 72076 Tübingen, Germany}
\author{Stefan Teufel}
\affiliation{Mathematisches Institut, University of Tübingen, Auf der Morgenstelle 10, 72076 Tübingen, Germany}

\begin{abstract}
We present an algorithm for reliably and systematically proving the existence of spectral gaps in Hamiltonians with quasicrystalline order, based on numerical calculations on finite domains. We apply this algorithm to prove that the Hofstadter model on the Ammann-Beenker tiling of the plane has spectral gaps at certain energies, and we are able to prove the existence of a spectral gap where previous numerical results were inconclusive. Our algorithm is applicable to more general systems with finite local complexity and eventually finds all gaps, circumventing an earlier no-go theorem regarding the computability of spectral gaps for general Hamiltonians.
\end{abstract}

\maketitle


The spectrum of a periodic Hamiltonian 
is characterized by bands and gaps or, e.g.\ in the case of incommensurable magnetic fields and two dimensions, by gaps only  (Cantor spectrum). 
Numerical results on finite patches, as well as experimental data, suggest that spectral gaps  appear also in systems with aperiodic order, such as the Hofstadter model on a quasicrystal \cite{Sire1989,Sire1990,benza_band_1991,zoorob2000,man2005,fuchs_hofstadter_2016,loring_bulk_2019}. Since Bloch theory is no longer available, alternative methods are needed to compute the spectrum of an infinitely extended aperiodic system. Using the method of quasi-modes (a.k.a.\ almost  eigenvectors), one can  deduce from the spectrum of the Hamiltonian restricted to a finite patch that certain intervals    \emph{contain}   spectrum of the infinite volume Hamiltonian $H$. 
See for example  \cite{colbrook_2019,trefethen2020spectra} for efficient implementations of this strategy  for aperiodic systems. Nevertheless,  crucial physical  features of such systems, such as topological phases \cite{Bandres2016,Kellendonk2019, Zilberberg2021},  depend on the existence and location 
 of \emph{spectral gaps}, that is the information that certain intervals do \emph{not contain} spectrum of the infinite volume Hamiltonian. However, until now, besides certain specific situations, the location of spectral gaps in aperiodic systems could only be guessed from numerical data on finite patches, but not conclusively confirmed.
 
In this letter, we present a  method that allows  to actually prove the existence of spectral gaps in infinite aperiodic systems based on numerical calculations on finite domains. Our method can be applied to any kind of short range Hamiltonian with finite local complexity. We apply it to the Hofstadter Hamiltonian  and to the $p_xp_y$-model, both on the Ammann-Beenker tiling. In particular, for the $p_xp_y$-model we prove that a small energy gap is open,  which appears in numerical calculations on finite balls, but whose existence for the infinite system was uncertain~\cite{loring_bulk_2019}. This demonstrates that a rigorous procedure like ours can be useful when numerical results are otherwise hard to interpret.

\begin{figure}[h]
\centering
\includegraphics{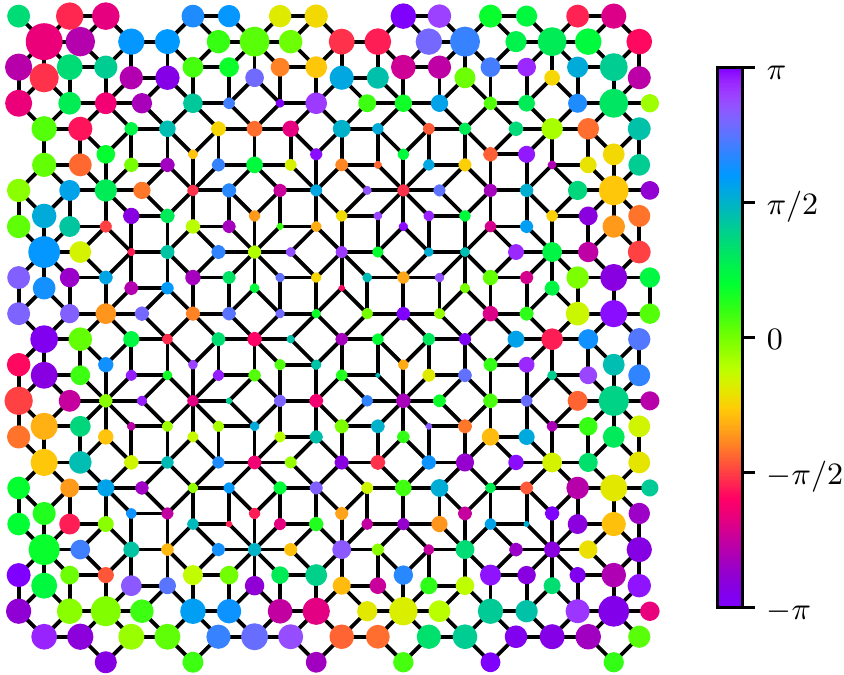}
\caption{\label{figure-edge-state}
An  edge state of a magnetic Hamiltonian on a finite patch of the Ammann-Beenker tiling ($L=9$). The color of the dots corresponds to the phase of the wave function and the radius  to the modulus. While the infinite system has a spectral gap at this energy, the finite system has an eigenstate whose mass is concentrated at the edges. 
}
\end{figure}
 
Numerical investigations of aperiodic systems necessarily restrict the Hamiltonian to finite patches and then either apply periodic boundary conditions, or use open boundary conditions. The first method relies on the construction of the so-called periodic approximant systems. While it has been shown that in specific models it is possible to successfully construct such periodic approximants \cite{Duneau1989,Tsunetsugu1991}, the mathematical justification of this procedure for generic aperiodic systems is still subject of ongoing works \cite{prodan2012,beckus2018,beckus2019,loring2020}. Furthermore, we stress that usually one does not have a precise control of the rate of convergence of the spectral properties of the periodic approximants. Using open boundary conditions, on the other hand, edge-states may appear in gaps of the \emph{bulk spectrum}, i.e.\ in the spectrum  of the infinite volume Hamiltonian. In previous works, such edge-states, as shown for example in Figure~\ref{figure-edge-state}, have been discarded as ``spectral pollution'' \cite{colbrook_2019, daviesPollution}.
 So far, however, no precise criterion has been formulated to decide whether a given state is an edge state and thus does not contribute to the bulk spectrum.
 Our  result formulated below yields such a criterion. Roughly speaking we show that if for some spectral window all quasi-modes  appearing in boxes of side-length $2L$   satisfy a quantifiable ``edge-state criterion'', then this spectral window is a gap in the spectrum of the infinite volume Hamiltonian. Our criterion can be numerically checked for finite-range Hamiltonians with finite local complexity, where for a given $L>0$ the restriction of $H$
 to a box of side-length $2L$ centred at an \emph{arbitrary} point in space yields a matrix from a \emph{finite} set of possible realizations.

We first formulate a general finite-size criterion which is equivalent to the existence of a bulk gap. For this we merely assume that the Hamiltonian $H$ is a bounded hermitian  operator on  $\ell^2(\Gamma;\mathcal{H})$,  where $\mathcal{H}$ is a separable Hilbert space and $\Gamma \subset \mathbb{R}^d$ is a countable set. We will also assume that there is a maximal distance $r > 0$ such that any point in $\mathbb{R}^d$ is within a distance at most $r$ of a point in~$\Gamma$.  Throughout this article we measure distances in $\mathbb{R}^d$ using the $\infty$-norm    $\|x\|_\infty := \max\{ |x_1|,\ldots, |x_d|\}$. We also write $B_L(x) := \prod_{j=1}^d (x_j-L, x_j+L)$ for the  open cube of side length $2L$ centered around the point $x \in \mathbb{R}^d$ and $\overline{B_L}(x)$ for its closure. For any set $A\subset \mathbb{R}^d$ and wave function $\psi\in \ell^2(\Gamma;\mathcal{H})$, we denote by  $\norm{\psi}_A^2 := \sum_{x\in \Gamma \cap A} \|\psi(x)\|_{\mathcal{H}}^2$ its $\ell^2$-mass  in~$A$.

\begin{definition}
\label{def-eps-bulk-gapped}
Let $\epsilon,L >0$ and $\lambda\in\mathbb{R}$. We say that a Hamiltonian $H$ is \textbf{$\epsilon$-bulk-gapped at energy $\lambda$ and scale $L$}, if there exist constants   $N \in \mathbb{N}$,   $N\geq 2$, and ${C<{1/N^{d}}}$, such that for any $x\in\Gamma$ and  any ${\psi \in \ell^2(\Gamma;\mathcal{H})}$ we have the following implication: Whenever
\begin{align}
\norm{ (H - \lambda) \psi}_ {B_L(x)} \leq \epsilon \norm{\psi}_{B_{L}(x)}\,,
\label{eq-prop}
\end{align}
then for $l := \frac{L + r}{N} + r$ it holds that
\begin{align}
\label{eq-prop-conclusio}
\norm{\psi}_{\overline{B_l}(x)}^2 \leq C \norm{\psi}_{B_L(x)}^2\,.
\end{align}
\end{definition}
In  words,  Definition~\ref{def-eps-bulk-gapped} requires that at some scale $L$ all $\epsilon$-quasi-modes \eqref{eq-prop} have an underproportional mass within the bulk \eqref{eq-prop-conclusio}. Although this property depends only on evaluations of the Hamiltonian $H$ on finite patches, we can show that it implies that the interval $(\lambda - \epsilon, \lambda + \epsilon)$ does not contain any spectrum of $H$.

\begin{theorem}\label{main-theorem}
If a Hamiltonian $H$ on $\ell^2(\Gamma; \mathcal{H})$ is locally $\epsilon$-bulk-gapped at energy $\lambda$ on some scale $L>0$, then the interval  $(\lambda - \epsilon, \lambda + \epsilon)$ is a gap in the spectrum of $H$, i.e.\
\begin{align*}
\sigma(H) \cap (\lambda - \epsilon, \lambda + \epsilon) = \varnothing.
\end{align*}
\end{theorem}
The  proof of Theorem~\ref{main-theorem} is given in the supplemental material,  Section~1.
To explain its simple geometric idea, suppose  that we are in dimension $d=1$, and  that the statement in Definition~\ref{def-eps-bulk-gapped} holds for $l= L/N$ instead of the more complicated expression involving $r$. We can then show that $H$ has no eigenvalues in $(\lambda - \epsilon, \lambda + \epsilon)$ using a simple decomposition argument:
 Define the lattices 
\begin{align*}
Z_q = 2L(\mathbb{Z} + q/N) \quad \text{for} \quad q \in \{1, \dots, N\} \,.
\end{align*}
While for fixed $q$ the   larger open intervals $B_L(x)$ centred at  different $x\in Z_q$   are still mutually  disjoint, the shorter closed intervals $\overline{B_l}(x)$ cover $\mathbb{R}$ when taking the union over all $q$, i.e.\
\begin{align}
\bigcup\nolimits_{q=1}^n \bigcup\nolimits_{x\in Z_q} \overline{B_l}(x)=\mathbb{R}\,.
\label{eq-cover}
\end{align}
Now assume that  $\psi$ is an eigenfunction of $H$ with eigenvalue   $\lambda_0\in (\lambda-\epsilon,\lambda+\epsilon)$. Then $\psi$  
 satisfies  \eqref{eq-prop} for any $x$, and thus, by assumption,  also \eqref{eq-prop-conclusio}. This implies
 \begin{align*}
\sum_{q=1}^N \sum_{x\in Z_q} \norm{\psi}^2_{\overline{B_l}(x)} < \frac{1}{N }\sum_{q=1}^N \sum_{x\in Z_q}  \norm{\psi}^2_{B_L(x)} =  \norm{\psi}^2\,,
\end{align*}
while \eqref{eq-cover} implies 
\begin{align*}
\sum_{q=1}^N \sum_{x\in Z_q} \norm{\psi}^2_{\overline{B_l}(x)} \geq \norm{\psi}^2\,.
\end{align*}
This is a contradiction and thus no such eigenfunction can exist. 

We now discuss how to establish existence of local $\epsilon$-bulk-gaps, and thus by Theorem~\ref{main-theorem} spectral gaps for Hamiltonians with finite range hoppings and finite local complexity.  
We first give a sufficient condition for a Hamiltonian to be locally $\epsilon$-bulk-gapped that can be efficiently verified numerically.
From now on we assume that  $H$ has only  finite range hoppings with maximal hopping length $m$, namely
$H_{xy} = 0$ for all $x,y\in\Gamma$ with  $\|x-y\|_\infty > m$. For any set $A\subset \mathbb{R}^d$ we denote by $\mathbf{1}_A$ the characteristic function of $A$ and set
 $H_A :=\mathbf{1}_A H \mathbf{1}_A$.
\begin{proposition}
\label{proposition}
As before, let $L > 0$, $N\in\mathbb{N}$,  $N\geq2$, ${l := \frac{L + r}{N} + r}$, and  $\lambda\in\mathbb{R}$. \\
Assume that for every $x \in \Gamma$   there exists   a set  $A(x)\subset \Gamma$ such that 
$\overline{B_l}(x)\subseteq A(x)\subseteq B_{L-m}(x)$,  $\lambda \notin \sigma(H_{A(x)})$, and
\begin{align*}
D(x) =\norm{\mathbf{1}_{\overline{B_l}(x)} (H_{A(x)} - \lambda)^{-1}  \mathbf{1}_{A(x)} H \mathbf{1}_{B_L(x) \backslash A(x)}}_\mathrm{op}
\end{align*}
satisfies $D(x)<N^{-d/2}$.\\
Then $H$ is $\epsilon$-bulk-gapped at energy $\lambda$ and scale $L$ for any~$\epsilon>0$ with 
\begin{align}
\epsilon < \inf_{x\in\Gamma} \frac{N^{-d/2} - D(x)}{\norm{\mathbf{1}_{\overline{B_l}(x)} (H_{A(x)}- \lambda)^{-1} \mathbf{1}_{A(x)}}_\mathrm{op}}\,.
\label{eq-conditions-epsilon}
\end{align}
 \end{proposition}

The   proof of Proposition~\ref{proposition} is given in the supplemental material, Section~1. Note that $H_{A(x)}$ is a sparse matrix and therefore the conditions of Proposition~\ref{proposition} can be efficiently checked using a sparse LU factorization \cite{superlu}, allowing computations for values of $L$ for which a direct diagonalization approach would no longer be computationally feasible.

\begin{figure}
\includegraphics{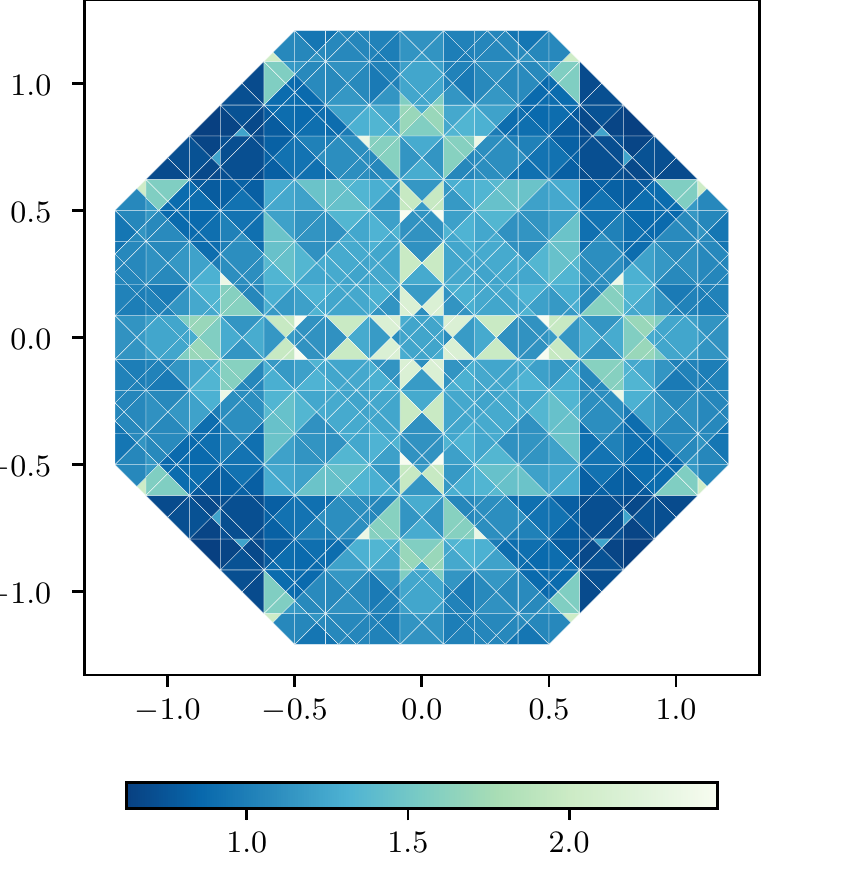}
\caption{\label{figure-decomposition}
The octagonal acceptance region for the Ammann-Beenker tiling decomposed into subpolygons corresponding to the different local patches in $\mathcal{C}_L(x)$, for $L = 5$. The color of each polygon corresponds to $D(x)$ computed for ${N=2}$ with $A(x) = B_{L - 1}(x)$ (in this case $m=1$), drawn in a logarithmic scale with base $10$,  for the Hofstadter Hamiltonian at magnetic field $b=1$ and energy $1.5$.  At scale $L = 5$, $D(x)$  never falls below the required bound $1/2$. At scale $L = 50$ one finds that $D(x)<1/2$ for every local patch, proving the~gap.}
\end{figure}

In order to apply Proposition~\ref{proposition} to quasicrystals $\Gamma\subset \mathbb{R}^d$, we still need to find a way to determine all \textit{local patches} at scale $L$,
\begin{align*}
\mathcal{C}_L(x) = \{ y\in \mathbb{R}^d \,|\, x+ y \in \Gamma\cap B_L(x)\} \subset B_L(0)\,,
\end{align*}
that occur for   $x\in \Gamma$. We say that $\Gamma$ has \textit{finite local complexity} if the set $\mathcal{C}_L= \{\mathcal{C}_L(x)\,|\, x\in\Gamma\}$ is finite \cite{lagarias1999,baake1998,lagarias_pleasants_2002}. Those quasiperiodic tilings that are suggested as models for physical systems all have finite local complexity. Our method is practical because, for common quasicrystals, the number of local patches  grows only polynomially in $L$, while it would grow exponentially without long-range order \cite{elser_comment_1985, moody_2000}.

To enumerate the tilings, we make use of the cut-and-project construction of quasiperiodic tilings \cite{debruijn_1981, janot_quasicrystals}. The Ammann-Beenker tiling \cite{ammann1992aperiodic, beenker1982algebraic} can be defined by a cut-and-project method using two ``projections'',
\begin{align*}
p = \begin{pmatrix}
1 & a & 0 & - a\\
0 & a & 1 & a
\end{pmatrix}\quad\mbox{and}\quad
\kappa = \begin{pmatrix}
1 & -a & 0 & a\\
0 & a & -1 & a
\end{pmatrix},
\end{align*}
both maps from $\mathbb{R}^4$ to $\mathbb{R}^2$, where $a = 1/\sqrt{2}$. We call $p$ and $\kappa$ ``projections'' since they represent orthogonal projections onto two orthogonal two-dimensional subspaces of $\mathbb{R}^4$ which are then both identified with $\mathbb{R}^2$. Furthermore, let $R\subset \mathbb{R}^2$ be the regular axis-aligned octagon centered at $0$ with side length $1$. The vertices of the Ammann-Beenker tiling are the set
\begin{align*}
\Gamma_\mathrm{AB} = \big\{\,p(z) \;\big|\; z \in \mathbb{Z}^4,\,\kappa(z) \in R\,\big\}\,.
\end{align*}
An edge is introduced between $p(z_1)$ and $p(z_2)$ whenever $z_1$ and $z_2$ differ by a unit vector in $\mathbb{R}^4$. Because $a$ is irrational, every point $x \in \Gamma_\mathrm{AB} $ has exactly one pre-image  $z \in \mathbb{Z}^4$ such that $p(z) = x$.

What local patches $\mathcal{C}_L(x)$ can occur  when $x$ varies in $\Gamma_\mathrm{AB} $? Let $z \in \mathbb{Z}^4$ be such that $p(z) = x\in \Gamma_\mathrm{AB}$, and consider the set of all $v \in \mathbb{Z}^4$ such that $p(z + v) \in \Gamma_\mathrm{AB} \cap B_L(x)$. There is only a finite number of $v \in \mathbb{Z}^4$ that can occur across all $x$, because $v$ has to fulfil the two conditions
\begin{align}\label{eq-cutproj-0}
p(z + v) \in B_L(x)\quad\text{and}\quad \kappa(z + v) \in R.
\end{align}
By the linearity of $p$,   the first condition reduces to
\begin{align}
\label{eq-cutproj-1}
p(v) \in B_L(0)
\end{align}
and since $\kappa(x) \in R$,   the second condition implies
\begin{align}
\label{eq-cutproj-2}
\kappa(v) \in 2R\,,
\end{align}
where   $2R$ is the octagon of side length 2.

Both \eqref{eq-cutproj-1} and \eqref{eq-cutproj-2} state that a linear image of $v$ lies in some compact set. Since $p$ and $\kappa$ have orthogonal kernels, these two conditions define a set $\tilde V_L$ linearly equivalent to the cartesian product of $B_L(0)$ and $2R$, which is itself compact and hence contains only finitely many integer points. Let $V_L = \tilde V_L \cap \mathbb{Z}^4$ be this set of ``candidate points''. An algorithm for computing $V_L$ is described in the supplemental material, Section~2, Algorithm~1.

According to \eqref{eq-cutproj-0}, for any $v\in V_L$ the patch $\mathcal{C}_L(x)$ contains the point $p(z+v)$ if and only if
\begin{align}
\kappa(z) \in R  - \kappa(v) \,,
\end{align}
i.e.\ if $\kappa(z)$ lies in the shifted octagon $R  - \kappa(v)$.
Put differently, every candidate point $v \in V_L$ decomposes the octagon $R$ into two disjoint sets $P^0(v)$ and $P^1(v)$, delineating for which $\kappa(z)$ the point $p(z + v)$ is or is not   in $\Gamma_\mathrm{AB}$. The enumeration algorithm then merely consists of computing all possible intersections
$
\bigcap_{v \in V} P^{i_v}(v)$, where $i\in\{0,1\}^{|V|}
$.
One can show that the number of such intersection that are nonempty, and hence the number of local patches, grows quadratically in $L$ for the Ammann-Beenker tiling \cite{julien2010}. They can be enumerated efficiently using a dynamic~programming approach described in Section~2, Algorithm~2 of the supplemental material. Figure~\ref{figure-decomposition} shows the resulting decomposition of the octagon $R$ for $L = 5$.

\begin{figure}[ht]
\centering
\includegraphics{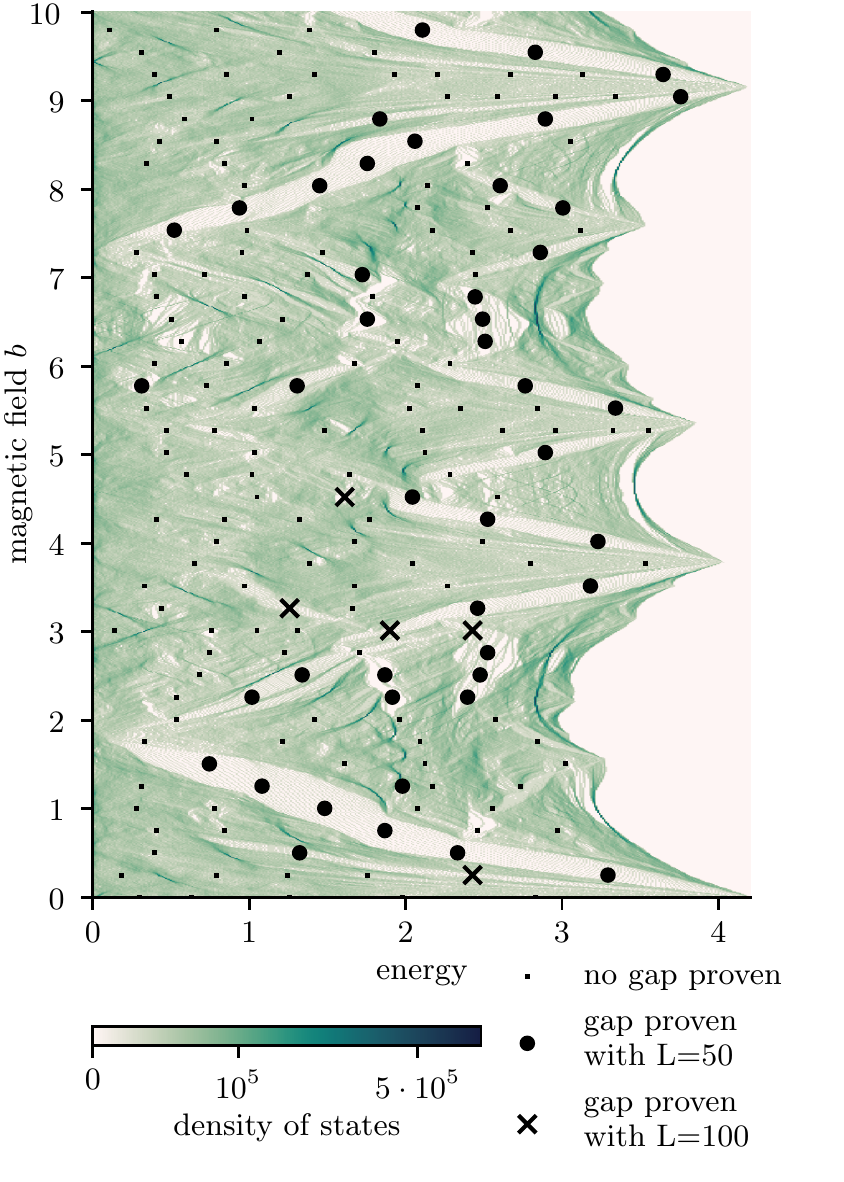}
\caption{\label{figure-butterfly}
The Hofstadter butterfly of the magnetic Laplacian on the Ammann-Beenker tiling with some points selected for investigation using our method.
}
\end{figure}

\begin{figure}[ht]
\centering
\includegraphics{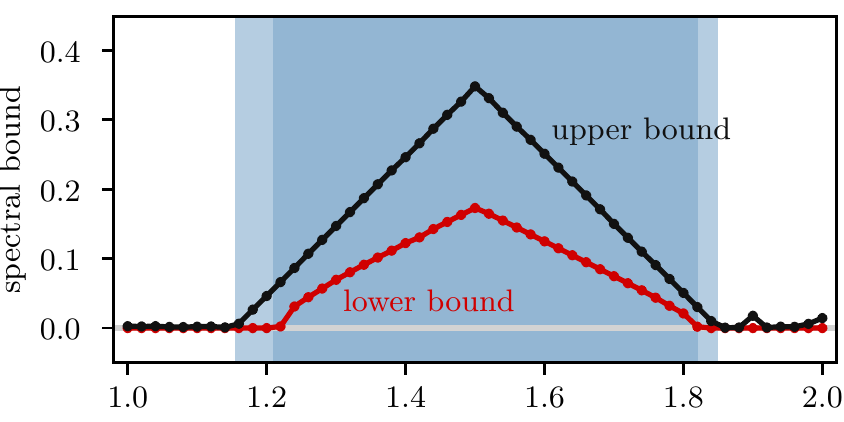}
\caption{\label{figure-section}
Comparison of our lower bound for the distance to spectrum (for $L = 50$) to the upper bound computed with the method of  \cite{colbrook_2019}. We computed both bounds for 50 equally spaced energies in the Hofstadter model at a constant magnetic field $b= 1$. The combination of these bounds allows us to bound the extent of the spectral gap containing energy $1.5$. The endpoints of the gap must be contained in the lighter shaded areas around energies $1.2$ and $1.82$. A similar plot for the Fibonacci crystal is shown in Figure~S1 in the supplemental material.
}
\end{figure}

We now study two physical systems on the Ammann-Beenker tiling. The $p_xp_y$-model is a model for a weak topological superconductor whose real-space description allows it to be defined on aperiodic sets \cite{fulga_aperiodic_2016}. The   matrix elements  of the Hamiltonian are
\begin{gather*}
H_{xy} = - t \sigma_3 - \tfrac{\mathrm{i}}{2} \Delta \sigma_1 \cos(\alpha_{xy}) - \tfrac{\mathrm{i}}{2} \Delta \sigma_2 \sin(\alpha_{xy})\,,\\
H_{xx} = - \mu \sigma_3\,,
\end{gather*}
for $\|x-y\|_2\leq1$, and $H_{xy}=0$ otherwise.
Here $\sigma_1, \sigma_2, \sigma_3$  are Pauli matrices, $\mu, \Delta \in \mathbb{R}$ and $\alpha_{xy}$ is the signed angle between the edge $xy$ and the $e_1$-axis.

For $\mu$ very large, this Hamiltonian can be considered as a small perturbation of $\mu\sigma_3$, thus its spectrum  has a  gap around 0, without edge states when restricted to  finite domains. As $\mu$  decreases, the gap eventually closes and the system is expected to undergo  a quantum phase transition into 
a topologically nontrivial phase.
This topological regime has been  studied in \cite{loring_bulk_2019}. Employing computational~$K$-theory, convincing evidence was found that a large  gap around zero  indeed  reopens. But the numerical data also suggested a second small gap might open. In the absence of a decisive criterion, the author had to leave open whether this gap persists in the thermodynamic limit \cite[p. 9]{loring_bulk_2019}.
Using our method, we could prove that there really is a small second gap around energy $0.804$ in the infinite system.

As a second example, we applied our method  to the Hofstadter model on the Ammann-Beenker tiling.  In the symmetric gauge, the matrix elements of the Hofstadter Hamiltonian are $
H_{xy} = \mathrm{e}^{\mathrm{i} b \det(x, y)}$
for $\|x-y\|_2\leq1$, and $H_{xy}=0$ otherwise, where $b\in \mathbb{R}$ denotes the strength of the magnetic field perpendicular to the tiling.
It was previously observed that patterns related to the Hofstadter butterfly also emerge in quasicrystalline systems \cite{vidal_quasiperiodic_2004, tran_topological_2015}. We approximated the density of states of the Hofstadter butterfly (see Figure~\ref{figure-butterfly}) by diagonalization of a finite system and created a set of possible gap locations by taking all local minima of a kernel density estimate with bandwidth $0.1$ of the spectra. In this way we generated 187 combinations of magnetic field and energy where a gap might be expected. Applying our algorithm with $L = 50$, we could show for 44 of these points that there is a spectral gap in the infinite system, increasing to 49 points with $L = 100$. For $L = 50$, this required checking 15,139 local patches, while for $L = 100$, we had $|\mathcal{C}_L| = 60{,}601$.

We also computed a cross-section of our bound at different energies for a fixed magnetic field (see Figure~\ref{figure-section}). Comparing our lower bound on the distance to spectrum to the upper bound in \cite{colbrook_2019} showed that the curves are similar and increase linearly towards the center of the gap. The combination of both bounds yields rigorous and precise information on the positions of the edge of the spectral gap. 
As discussed in the Section~4 of the supplemental material, one can indeed deduce from our result that the problem of determining the spectrum of an operator with finite local complexity has a \emph{solvability complexity index} \cite{hansen2011} equal to one.
In particular, this means that for any energy $\lambda$ that is not in the spectrum, one can determine an $L$ such that our algorithm applied on the scale $L$ proves that some neighborhood of $\lambda$ is contained in a gap. Together with the convergence result of \cite{colbrook_2019}, this means that the gap estimate in Figure~\ref{figure-section} can be made arbitrarily accurate by choosing $L$ large enough.

In conclusion, we have described a new way to prove the existence of spectral gaps in infinitely extended quasicrystalline systems based on numerical computations on finite patches. The algorithm exploits the fact that in systems of finite local complexity, the spectral gaps of the bulk operator can be derived from the structure of the eigenstates of its restrictions to patches of a fixed finite size. In this way, we circumvent the no-go theorem of \cite{colbrook_2019}, which states that there is no algorithm that computes spectral gaps in general infinitely extended systems.


\bibliography{quasicrystals}

\begin{thebibliography}{34}%
\makeatletter
\providecommand \@ifxundefined [1]{%
 \@ifx{#1\undefined}
}%
\providecommand \@ifnum [1]{%
 \ifnum #1\expandafter \@firstoftwo
 \else \expandafter \@secondoftwo
 \fi
}%
\providecommand \@ifx [1]{%
 \ifx #1\expandafter \@firstoftwo
 \else \expandafter \@secondoftwo
 \fi
}%
\providecommand \natexlab [1]{#1}%
\providecommand \enquote  [1]{``#1''}%
\providecommand \bibnamefont  [1]{#1}%
\providecommand \bibfnamefont [1]{#1}%
\providecommand \citenamefont [1]{#1}%
\providecommand \href@noop [0]{\@secondoftwo}%
\providecommand \href [0]{\begingroup \@sanitize@url \@href}%
\providecommand \@href[1]{\@@startlink{#1}\@@href}%
\providecommand \@@href[1]{\endgroup#1\@@endlink}%
\providecommand \@sanitize@url [0]{\catcode `\\12\catcode `\$12\catcode
  `\&12\catcode `\#12\catcode `\^12\catcode `\_12\catcode `\%12\relax}%
\providecommand \@@startlink[1]{}%
\providecommand \@@endlink[0]{}%
\providecommand \url  [0]{\begingroup\@sanitize@url \@url }%
\providecommand \@url [1]{\endgroup\@href {#1}{\urlprefix }}%
\providecommand \urlprefix  [0]{URL }%
\providecommand \Eprint [0]{\href }%
\providecommand \doibase [0]{https://doi.org/}%
\providecommand \selectlanguage [0]{\@gobble}%
\providecommand \bibinfo  [0]{\@secondoftwo}%
\providecommand \bibfield  [0]{\@secondoftwo}%
\providecommand \translation [1]{[#1]}%
\providecommand \BibitemOpen [0]{}%
\providecommand \bibitemStop [0]{}%
\providecommand \bibitemNoStop [0]{.\EOS\space}%
\providecommand \EOS [0]{\spacefactor3000\relax}%
\providecommand \BibitemShut  [1]{\csname bibitem#1\endcsname}%
\let\auto@bib@innerbib\@empty
\bibitem [{\citenamefont {Sire}(1989)}]{Sire1989}%
  \BibitemOpen
  \bibfield  {author} {\bibinfo {author} {\bibfnamefont {C.}~\bibnamefont
  {Sire}},\ }\bibfield  {title} {\bibinfo {title} {Electronic spectrum of a 2d
  quasi-crystal related to the octagonal quasi-periodic tiling},\ }\href@noop
  {} {\bibfield  {journal} {\bibinfo  {journal} {EPL}\ }\textbf {\bibinfo
  {volume} {10}},\ \bibinfo {pages} {483} (\bibinfo {year} {1989})}\BibitemShut
  {NoStop}%
\bibitem [{\citenamefont {Sire}\ and\ \citenamefont
  {Bellissard}(1990)}]{Sire1990}%
  \BibitemOpen
  \bibfield  {author} {\bibinfo {author} {\bibfnamefont {C.}~\bibnamefont
  {Sire}}\ and\ \bibinfo {author} {\bibfnamefont {J.}~\bibnamefont
  {Bellissard}},\ }\bibfield  {title} {\bibinfo {title} {Renormalization group
  for the octagonal quasi-periodic tiling},\ }\href@noop {} {\bibfield
  {journal} {\bibinfo  {journal} {EPL}\ }\textbf {\bibinfo {volume} {11}},\
  \bibinfo {pages} {439} (\bibinfo {year} {1990})}\BibitemShut {NoStop}%
\bibitem [{\citenamefont {Benza}\ and\ \citenamefont
  {Sire}(1991)}]{benza_band_1991}%
  \BibitemOpen
  \bibfield  {author} {\bibinfo {author} {\bibfnamefont {V.~G.}\ \bibnamefont
  {Benza}}\ and\ \bibinfo {author} {\bibfnamefont {C.}~\bibnamefont {Sire}},\
  }\bibfield  {title} {\bibinfo {title} {Band spectrum of the octagonal
  quasicrystal: {Finite} measure, gaps, and chaos},\ }\href
  {https://doi.org/10.1103/PhysRevB.44.10343} {\bibfield  {journal} {\bibinfo
  {journal} {Phys.\ Rev.\ B}\ }\textbf {\bibinfo {volume} {44}},\ \bibinfo
  {pages} {10343} (\bibinfo {year} {1991})}\BibitemShut {NoStop}%
\bibitem [{\citenamefont {Zoorob}\ \emph {et~al.}(2000)\citenamefont {Zoorob},
  \citenamefont {Charlton}, \citenamefont {Parker}, \citenamefont {Baumberg},\
  and\ \citenamefont {Netti}}]{zoorob2000}%
  \BibitemOpen
  \bibfield  {author} {\bibinfo {author} {\bibfnamefont {M.~E.}\ \bibnamefont
  {Zoorob}}, \bibinfo {author} {\bibfnamefont {M.~D.~B.}\ \bibnamefont
  {Charlton}}, \bibinfo {author} {\bibfnamefont {G.~J.}\ \bibnamefont
  {Parker}}, \bibinfo {author} {\bibfnamefont {J.~J.}\ \bibnamefont
  {Baumberg}},\ and\ \bibinfo {author} {\bibfnamefont {M.~C.}\ \bibnamefont
  {Netti}},\ }\bibfield  {title} {\bibinfo {title} {Complete photonic bandgaps
  in 12-fold symmetric quasicrystals},\ }\href@noop {} {\bibfield  {journal}
  {\bibinfo  {journal} {Nature}\ }\textbf {\bibinfo {volume} {404}},\ \bibinfo
  {pages} {740} (\bibinfo {year} {2000})}\BibitemShut {NoStop}%
\bibitem [{\citenamefont {Man}\ \emph {et~al.}(2005)\citenamefont {Man},
  \citenamefont {Megens}, \citenamefont {Steinhardt},\ and\ \citenamefont
  {Chaikin}}]{man2005}%
  \BibitemOpen
  \bibfield  {author} {\bibinfo {author} {\bibfnamefont {W.}~\bibnamefont
  {Man}}, \bibinfo {author} {\bibfnamefont {M.}~\bibnamefont {Megens}},
  \bibinfo {author} {\bibfnamefont {P.~J.}\ \bibnamefont {Steinhardt}},\ and\
  \bibinfo {author} {\bibfnamefont {P.~M.}\ \bibnamefont {Chaikin}},\
  }\bibfield  {title} {\bibinfo {title} {Experimental measurement of the
  photonic properties of icosahedral quasicrystals},\ }\href@noop {} {\bibfield
   {journal} {\bibinfo  {journal} {Nature}\ }\textbf {\bibinfo {volume}
  {436}},\ \bibinfo {pages} {993} (\bibinfo {year} {2005})}\BibitemShut
  {NoStop}%
\bibitem [{\citenamefont {Fuchs}\ and\ \citenamefont
  {Vidal}(2016)}]{fuchs_hofstadter_2016}%
  \BibitemOpen
  \bibfield  {author} {\bibinfo {author} {\bibfnamefont {J.-N.}\ \bibnamefont
  {Fuchs}}\ and\ \bibinfo {author} {\bibfnamefont {J.}~\bibnamefont {Vidal}},\
  }\bibfield  {title} {\bibinfo {title} {Hofstadter butterfly of a
  quasicrystal},\ }\href {https://doi.org/10.1103/PhysRevB.94.205437}
  {\bibfield  {journal} {\bibinfo  {journal} {Phys.\ Rev.\ B}\ }\textbf
  {\bibinfo {volume} {94}},\ \bibinfo {pages} {205437} (\bibinfo {year}
  {2016})}\BibitemShut {NoStop}%
\bibitem [{\citenamefont {Loring}(2019)}]{loring_bulk_2019}%
  \BibitemOpen
  \bibfield  {author} {\bibinfo {author} {\bibfnamefont {T.~A.}\ \bibnamefont
  {Loring}},\ }\bibfield  {title} {\bibinfo {title} {Bulk spectrum and
  {K}-theory for infinite-area topological quasicrystals},\ }\href
  {https://doi.org/10.1063/1.5083051} {\bibfield  {journal} {\bibinfo
  {journal} {J.\ Math.\ Phys.}\ }\textbf {\bibinfo {volume} {60}},\ \bibinfo
  {pages} {081903} (\bibinfo {year} {2019})}\BibitemShut {NoStop}%
\bibitem [{\citenamefont {Colbrook}\ \emph {et~al.}(2019)\citenamefont
  {Colbrook}, \citenamefont {Roman},\ and\ \citenamefont
  {Hansen}}]{colbrook_2019}%
  \BibitemOpen
  \bibfield  {author} {\bibinfo {author} {\bibfnamefont {M.~J.}\ \bibnamefont
  {Colbrook}}, \bibinfo {author} {\bibfnamefont {B.}~\bibnamefont {Roman}},\
  and\ \bibinfo {author} {\bibfnamefont {A.~C.}\ \bibnamefont {Hansen}},\
  }\bibfield  {title} {\bibinfo {title} {{How to Compute Spectra with Error
  Control}},\ }\href {https://doi.org/10.1103/PhysRevLett.122.250201}
  {\bibfield  {journal} {\bibinfo  {journal} {Phys.\ Rev.\ Lett.}\ }\textbf
  {\bibinfo {volume} {122}},\ \bibinfo {pages} {250201} (\bibinfo {year}
  {2019})}\BibitemShut {NoStop}%
\bibitem [{\citenamefont {Trefethen}\ and\ \citenamefont
  {Embree}(2020)}]{trefethen2020spectra}%
  \BibitemOpen
  \bibfield  {author} {\bibinfo {author} {\bibfnamefont {L.~N.}\ \bibnamefont
  {Trefethen}}\ and\ \bibinfo {author} {\bibfnamefont {M.}~\bibnamefont
  {Embree}},\ }\href@noop {} {\emph {\bibinfo {title} {Spectra and
  pseudospectra}}}\ (\bibinfo  {publisher} {Princeton University Press},\
  \bibinfo {year} {2020})\BibitemShut {NoStop}%
\bibitem [{\citenamefont {Bandres}\ \emph {et~al.}(2016)\citenamefont
  {Bandres}, \citenamefont {Rechtsman},\ and\ \citenamefont
  {Segev}}]{Bandres2016}%
  \BibitemOpen
  \bibfield  {author} {\bibinfo {author} {\bibfnamefont {M.~A.}\ \bibnamefont
  {Bandres}}, \bibinfo {author} {\bibfnamefont {M.~C.}\ \bibnamefont
  {Rechtsman}},\ and\ \bibinfo {author} {\bibfnamefont {M.}~\bibnamefont
  {Segev}},\ }\bibfield  {title} {\bibinfo {title} {Topological photonic
  quasicrystals: fractal topological spectrum and protected transport},\
  }\href@noop {} {\bibfield  {journal} {\bibinfo  {journal} {Phys. Rev. X}\
  }\textbf {\bibinfo {volume} {6}} (\bibinfo {year} {2016})}\BibitemShut
  {NoStop}%
\bibitem [{\citenamefont {Kellendonk}\ and\ \citenamefont
  {Prodan}(2019)}]{Kellendonk2019}%
  \BibitemOpen
  \bibfield  {author} {\bibinfo {author} {\bibfnamefont {J.}~\bibnamefont
  {Kellendonk}}\ and\ \bibinfo {author} {\bibfnamefont {E.}~\bibnamefont
  {Prodan}},\ }\bibfield  {title} {\bibinfo {title} {Bulk{\textendash}boundary
  correspondence for sturmian kohmoto-like models},\ }\href@noop {} {\bibfield
  {journal} {\bibinfo  {journal} {Ann. Henri Poincar{\'e}}\ }\textbf {\bibinfo
  {volume} {20}},\ \bibinfo {pages} {2039} (\bibinfo {year}
  {2019})}\BibitemShut {NoStop}%
\bibitem [{\citenamefont {Zilberberg}(2021)}]{Zilberberg2021}%
  \BibitemOpen
  \bibfield  {author} {\bibinfo {author} {\bibfnamefont {O.}~\bibnamefont
  {Zilberberg}},\ }\bibfield  {title} {\bibinfo {title} {Topology in
  quasicrystals},\ }\href@noop {} {\bibfield  {journal} {\bibinfo  {journal}
  {Opt. Mater. Express}\ }\textbf {\bibinfo {volume} {11}},\ \bibinfo {pages}
  {1143} (\bibinfo {year} {2021})}\BibitemShut {NoStop}%
\bibitem [{\citenamefont {Duneau}(1989)}]{Duneau1989}%
  \BibitemOpen
  \bibfield  {author} {\bibinfo {author} {\bibfnamefont {M.}~\bibnamefont
  {Duneau}},\ }\bibfield  {title} {\bibinfo {title} {Approximants of
  quasiperiodic structures generated by the inflation mapping},\ }\href@noop {}
  {\bibfield  {journal} {\bibinfo  {journal} {J. Phys. A}\ }\textbf {\bibinfo
  {volume} {22}},\ \bibinfo {pages} {4549} (\bibinfo {year}
  {1989})}\BibitemShut {NoStop}%
\bibitem [{\citenamefont {Tsunetsugu}\ \emph {et~al.}(1991)\citenamefont
  {Tsunetsugu}, \citenamefont {Fujiwara}, \citenamefont {Ueda},\ and\
  \citenamefont {Tokihiro}}]{Tsunetsugu1991}%
  \BibitemOpen
  \bibfield  {author} {\bibinfo {author} {\bibfnamefont {H.}~\bibnamefont
  {Tsunetsugu}}, \bibinfo {author} {\bibfnamefont {T.}~\bibnamefont
  {Fujiwara}}, \bibinfo {author} {\bibfnamefont {K.}~\bibnamefont {Ueda}},\
  and\ \bibinfo {author} {\bibfnamefont {T.}~\bibnamefont {Tokihiro}},\
  }\bibfield  {title} {\bibinfo {title} {Electronic properties of the penrose
  lattice. i. energy spectrum and wave functions},\ }\href@noop {} {\bibfield
  {journal} {\bibinfo  {journal} {Phys. Rev. B}\ }\textbf {\bibinfo {volume}
  {43}},\ \bibinfo {pages} {8879} (\bibinfo {year} {1991})}\BibitemShut
  {NoStop}%
\bibitem [{\citenamefont {Prodan}(2012)}]{prodan2012}%
  \BibitemOpen
  \bibfield  {author} {\bibinfo {author} {\bibfnamefont {E.}~\bibnamefont
  {Prodan}},\ }\bibfield  {title} {\bibinfo {title} {Quantum transport in
  disordered systems under magnetic fields: a study based on operator
  algebras},\ }\href@noop {} {\bibfield  {journal} {\bibinfo  {journal} {Appl.
  Math. Res. Express}\ } (\bibinfo {year} {2012})}\BibitemShut {NoStop}%
\bibitem [{\citenamefont {Beckus}\ \emph {et~al.}(2018)\citenamefont {Beckus},
  \citenamefont {Bellissard},\ and\ \citenamefont {De~Nittis}}]{beckus2018}%
  \BibitemOpen
  \bibfield  {author} {\bibinfo {author} {\bibfnamefont {S.}~\bibnamefont
  {Beckus}}, \bibinfo {author} {\bibfnamefont {J.}~\bibnamefont {Bellissard}},\
  and\ \bibinfo {author} {\bibfnamefont {G.}~\bibnamefont {De~Nittis}},\
  }\bibfield  {title} {\bibinfo {title} {Spectral {C}ontinuity for {A}periodic
  {Q}uantum {S}ystems {I}. {G}eneral {T}heory},\ }\href@noop {} {\bibfield
  {journal} {\bibinfo  {journal} {J. Funct. Anal.}\ }\textbf {\bibinfo {volume}
  {275}},\ \bibinfo {pages} {2917} (\bibinfo {year} {2018})}\BibitemShut
  {NoStop}%
\bibitem [{\citenamefont {Beckus}\ \emph {et~al.}(2019)\citenamefont {Beckus},
  \citenamefont {Bellissard},\ and\ \citenamefont {Cornean}}]{beckus2019}%
  \BibitemOpen
  \bibfield  {author} {\bibinfo {author} {\bibfnamefont {S.}~\bibnamefont
  {Beckus}}, \bibinfo {author} {\bibfnamefont {J.}~\bibnamefont {Bellissard}},\
  and\ \bibinfo {author} {\bibfnamefont {H.}~\bibnamefont {Cornean}},\
  }\bibfield  {title} {\bibinfo {title} {{H}{\"o}lder {C}ontinuity of the
  {S}pectra for {A}periodic {H}amiltonians},\ }\href@noop {} {\bibfield
  {journal} {\bibinfo  {journal} {Ann. Henri Poincar{\'e}}\ }\textbf {\bibinfo
  {volume} {20}},\ \bibinfo {pages} {3603} (\bibinfo {year}
  {2019})}\BibitemShut {NoStop}%
\bibitem [{\citenamefont {Loring}\ and\ \citenamefont
  {Schulz-Baldes}(2020)}]{loring2020}%
  \BibitemOpen
  \bibfield  {author} {\bibinfo {author} {\bibfnamefont {T.~A.}\ \bibnamefont
  {Loring}}\ and\ \bibinfo {author} {\bibfnamefont {H.}~\bibnamefont
  {Schulz-Baldes}},\ }\bibfield  {title} {\bibinfo {title} {The spectral
  localizer for even index pairings},\ }\href@noop {} {\bibfield  {journal}
  {\bibinfo  {journal} {J. Noncommutative Geom.}\ }\textbf {\bibinfo {volume}
  {14}},\ \bibinfo {pages} {1} (\bibinfo {year} {2020})}\BibitemShut {NoStop}%
\bibitem [{\citenamefont {Davies}(2004)}]{daviesPollution}%
  \BibitemOpen
  \bibfield  {author} {\bibinfo {author} {\bibfnamefont {E.~B.}\ \bibnamefont
  {Davies}},\ }\bibfield  {title} {\bibinfo {title} {Spectral pollution},\
  }\href {https://doi.org/10.1093/imanum/24.3.417} {\bibfield  {journal}
  {\bibinfo  {journal} {IMA J.\ Numer.\ Anal.}\ }\textbf {\bibinfo {volume}
  {24}},\ \bibinfo {pages} {417} (\bibinfo {year} {2004})}\BibitemShut
  {NoStop}%
\bibitem [{\citenamefont {Li}(2005)}]{superlu}%
  \BibitemOpen
  \bibfield  {author} {\bibinfo {author} {\bibfnamefont {X.~S.}\ \bibnamefont
  {Li}},\ }\bibfield  {title} {\bibinfo {title} {An {O}verview of {SuperLU}:
  {A}lgorithms, {I}mplementation, and {U}ser {I}nterface},\ }\href@noop {}
  {\bibfield  {journal} {\bibinfo  {journal} {ACM T.\ Math.\ Software}\
  }\textbf {\bibinfo {volume} {31}},\ \bibinfo {pages} {302} (\bibinfo {year}
  {2005})}\BibitemShut {NoStop}%
\bibitem [{\citenamefont {Lagarias}(1999)}]{lagarias1999}%
  \BibitemOpen
  \bibfield  {author} {\bibinfo {author} {\bibfnamefont {J.~C.}\ \bibnamefont
  {Lagarias}},\ }\bibfield  {title} {\bibinfo {title} {{G}eometric models for
  {Q}uasicrystals {I}. {D}elone sets of finite type},\ }\href@noop {}
  {\bibfield  {journal} {\bibinfo  {journal} {Discrete Comput. Geom.}\ }\textbf
  {\bibinfo {volume} {21}},\ \bibinfo {pages} {161} (\bibinfo {year}
  {1999})}\BibitemShut {NoStop}%
\bibitem [{\citenamefont {Baake}\ and\ \citenamefont
  {Moody}(1998)}]{baake1998}%
  \BibitemOpen
  \bibfield  {author} {\bibinfo {author} {\bibfnamefont {M.}~\bibnamefont
  {Baake}}\ and\ \bibinfo {author} {\bibfnamefont {R.~V.}\ \bibnamefont
  {Moody}},\ }\bibfield  {title} {\bibinfo {title} {Diffractive point sets with
  entropy},\ }\href@noop {} {\bibfield  {journal} {\bibinfo  {journal} {J.\
  Phys.\ A - Math.\ Gen.}\ }\textbf {\bibinfo {volume} {31}},\ \bibinfo {pages}
  {9023} (\bibinfo {year} {1998})}\BibitemShut {NoStop}%
\bibitem [{\citenamefont {Lagarias}\ and\ \citenamefont
  {Pleasants}(2002)}]{lagarias_pleasants_2002}%
  \BibitemOpen
  \bibfield  {author} {\bibinfo {author} {\bibfnamefont {J.~C.}\ \bibnamefont
  {Lagarias}}\ and\ \bibinfo {author} {\bibfnamefont {P.~A.~B.}\ \bibnamefont
  {Pleasants}},\ }\bibfield  {title} {\bibinfo {title} {{L}ocal {C}omplexity of
  {D}elone {S}ets and {C}rystallinity},\ }\href
  {https://doi.org/10.4153/CMB-2002-058-0} {\bibfield  {journal} {\bibinfo
  {journal} {Can.\ Math.\ Bull.}\ }\textbf {\bibinfo {volume} {45}},\ \bibinfo
  {pages} {634–652} (\bibinfo {year} {2002})}\BibitemShut {NoStop}%
\bibitem [{\citenamefont {Elser}(1985)}]{elser_comment_1985}%
  \BibitemOpen
  \bibfield  {author} {\bibinfo {author} {\bibfnamefont {V.}~\bibnamefont
  {Elser}},\ }\bibfield  {title} {\bibinfo {title} {Comment on
  ``{Q}uasicrystals: A new class of ordered structures''},\ }\href
  {https://doi.org/10.1103/PhysRevLett.54.1730} {\bibfield  {journal} {\bibinfo
   {journal} {Phys.\ Rev.\ Lett.}\ }\textbf {\bibinfo {volume} {54}},\ \bibinfo
  {pages} {1730} (\bibinfo {year} {1985})}\BibitemShut {NoStop}%
\bibitem [{\citenamefont {Moody}(2000)}]{moody_2000}%
  \BibitemOpen
  \bibfield  {author} {\bibinfo {author} {\bibfnamefont {R.~V.}\ \bibnamefont
  {Moody}},\ }\bibfield  {title} {\bibinfo {title} {Model {S}ets: {A}
  {S}urvey},\ }in\ \href@noop {} {\emph {\bibinfo {booktitle} {From
  quasicrystals to more complex systems}}}\ (\bibinfo  {publisher} {Springer},\
  \bibinfo {year} {2000})\ pp.\ \bibinfo {pages} {145--166}\BibitemShut
  {NoStop}%
\bibitem [{\citenamefont {de~Bruijn}(1981)}]{debruijn_1981}%
  \BibitemOpen
  \bibfield  {author} {\bibinfo {author} {\bibfnamefont {N.~G.}\ \bibnamefont
  {de~Bruijn}},\ }\bibfield  {title} {\bibinfo {title} {Algebraic theory of
  {P}enrose’s non-periodic tilings of the plane},\ }\href@noop {} {\bibfield
  {journal} {\bibinfo  {journal} {Indag.\ Math.}\ }\textbf {\bibinfo {volume}
  {84}},\ \bibinfo {pages} {39} (\bibinfo {year} {1981})}\BibitemShut {NoStop}%
\bibitem [{\citenamefont {Janot}(2012)}]{janot_quasicrystals}%
  \BibitemOpen
  \bibfield  {author} {\bibinfo {author} {\bibfnamefont {C.}~\bibnamefont
  {Janot}},\ }\href@noop {} {\emph {\bibinfo {title} {Quasicrystals:~A
  Primer}}},\ \bibinfo {edition} {2nd}\ ed.\ (\bibinfo  {publisher} {Oxford
  University Press},\ \bibinfo {year} {2012})\BibitemShut {NoStop}%
\bibitem [{\citenamefont {Ammann}\ \emph {et~al.}(1992)\citenamefont {Ammann},
  \citenamefont {Gr{\"u}nbaum},\ and\ \citenamefont
  {Shephard}}]{ammann1992aperiodic}%
  \BibitemOpen
  \bibfield  {author} {\bibinfo {author} {\bibfnamefont {R.}~\bibnamefont
  {Ammann}}, \bibinfo {author} {\bibfnamefont {B.}~\bibnamefont
  {Gr{\"u}nbaum}},\ and\ \bibinfo {author} {\bibfnamefont {G.~C.}\ \bibnamefont
  {Shephard}},\ }\bibfield  {title} {\bibinfo {title} {Aperiodic tiles},\
  }\href@noop {} {\bibfield  {journal} {\bibinfo  {journal} {Discrete Comput.\
  Geom.}\ }\textbf {\bibinfo {volume} {8}},\ \bibinfo {pages} {1} (\bibinfo
  {year} {1992})}\BibitemShut {NoStop}%
\bibitem [{\citenamefont {Beenker}(1982)}]{beenker1982algebraic}%
  \BibitemOpen
  \bibfield  {author} {\bibinfo {author} {\bibfnamefont {F.~P.~M.}\
  \bibnamefont {Beenker}},\ }\href@noop {} {\emph {\bibinfo {title} {Algebraic
  theory of non-periodic tilings of the plane by two simple building blocks: a
  square and a rhombus}}},\ Technical Report 82-WSK-04\ (\bibinfo  {publisher}
  {Eindhoven University of Technology},\ \bibinfo {year} {1982})\BibitemShut
  {NoStop}%
\bibitem [{\citenamefont {Julien}(2010)}]{julien2010}%
  \BibitemOpen
  \bibfield  {author} {\bibinfo {author} {\bibfnamefont {A.}~\bibnamefont
  {Julien}},\ }\bibfield  {title} {\bibinfo {title} {Complexity and
  {C}ohomology for {C}ut-and-{P}rojection {T}ilings},\ }\href@noop {}
  {\bibfield  {journal} {\bibinfo  {journal} {Ergodic Theory and Dynamical
  Systems}\ }\textbf {\bibinfo {volume} {30}},\ \bibinfo {pages} {489}
  (\bibinfo {year} {2010})}\BibitemShut {NoStop}%
\bibitem [{\citenamefont {Fulga}\ \emph {et~al.}(2016)\citenamefont {Fulga},
  \citenamefont {Pikulin},\ and\ \citenamefont
  {Loring}}]{fulga_aperiodic_2016}%
  \BibitemOpen
  \bibfield  {author} {\bibinfo {author} {\bibfnamefont {I.~C.}\ \bibnamefont
  {Fulga}}, \bibinfo {author} {\bibfnamefont {D.~I.}\ \bibnamefont {Pikulin}},\
  and\ \bibinfo {author} {\bibfnamefont {T.~A.}\ \bibnamefont {Loring}},\
  }\bibfield  {title} {\bibinfo {title} {Aperiodic weak topological
  superconductors},\ }\href@noop {} {\bibfield  {journal} {\bibinfo  {journal}
  {Phys.\ Rev.\ Lett.}\ }\textbf {\bibinfo {volume} {116}},\ \bibinfo {pages}
  {257002} (\bibinfo {year} {2016})}\BibitemShut {NoStop}%
\bibitem [{\citenamefont {Vidal}\ and\ \citenamefont
  {Mosseri}(2004)}]{vidal_quasiperiodic_2004}%
  \BibitemOpen
  \bibfield  {author} {\bibinfo {author} {\bibfnamefont {J.}~\bibnamefont
  {Vidal}}\ and\ \bibinfo {author} {\bibfnamefont {R.}~\bibnamefont
  {Mosseri}},\ }\bibfield  {title} {\bibinfo {title} {Quasiperiodic tilings in
  a magnetic field},\ }\href {https://doi.org/10.1016/j.jnoncrysol.2003.11.027}
  {\bibfield  {journal} {\bibinfo  {journal} {J. Non-Cryst. Solids}\ }\textbf
  {\bibinfo {volume} {334--335}},\ \bibinfo {pages} {130} (\bibinfo {year}
  {2004})}\BibitemShut {NoStop}%
\bibitem [{\citenamefont {Tran}\ \emph {et~al.}(2015)\citenamefont {Tran},
  \citenamefont {Dauphin}, \citenamefont {Goldman},\ and\ \citenamefont
  {Gaspard}}]{tran_topological_2015}%
  \BibitemOpen
  \bibfield  {author} {\bibinfo {author} {\bibfnamefont {D.-T.}\ \bibnamefont
  {Tran}}, \bibinfo {author} {\bibfnamefont {A.}~\bibnamefont {Dauphin}},
  \bibinfo {author} {\bibfnamefont {N.}~\bibnamefont {Goldman}},\ and\ \bibinfo
  {author} {\bibfnamefont {P.}~\bibnamefont {Gaspard}},\ }\bibfield  {title}
  {\bibinfo {title} {Topological {Hofstadter} insulators in a two-dimensional
  quasicrystal},\ }\href {https://doi.org/10.1103/PhysRevB.91.085125}
  {\bibfield  {journal} {\bibinfo  {journal} {Phys.\ Rev.\ B}\ }\textbf
  {\bibinfo {volume} {91}},\ \bibinfo {pages} {085125} (\bibinfo {year}
  {2015})}\BibitemShut {NoStop}%
\bibitem [{\citenamefont {Hansen}(2011)}]{hansen2011}%
  \BibitemOpen
  \bibfield  {author} {\bibinfo {author} {\bibfnamefont {A.}~\bibnamefont
  {Hansen}},\ }\bibfield  {title} {\bibinfo {title} {On the solvability
  complexity index, the n-pseudospectrum and approximations of spectra of
  operators},\ }\href@noop {} {\bibfield  {journal} {\bibinfo  {journal} {J.\
  Am.\ Math.\ Soc.}\ }\textbf {\bibinfo {volume} {24}},\ \bibinfo {pages} {81}
  (\bibinfo {year} {2011})}\BibitemShut {NoStop}%
\end{thebibliography}%


\begin{thebibliography}{16}%
\makeatletter
\providecommand \@ifxundefined [1]{%
 \@ifx{#1\undefined}
}%
\providecommand \@ifnum [1]{%
 \ifnum #1\expandafter \@firstoftwo
 \else \expandafter \@secondoftwo
 \fi
}%
\providecommand \@ifx [1]{%
 \ifx #1\expandafter \@firstoftwo
 \else \expandafter \@secondoftwo
 \fi
}%
\providecommand \natexlab [1]{#1}%
\providecommand \enquote  [1]{``#1''}%
\providecommand \bibnamefont  [1]{#1}%
\providecommand \bibfnamefont [1]{#1}%
\providecommand \citenamefont [1]{#1}%
\providecommand \href@noop [0]{\@secondoftwo}%
\providecommand \href [0]{\begingroup \@sanitize@url \@href}%
\providecommand \@href[1]{\@@startlink{#1}\@@href}%
\providecommand \@@href[1]{\endgroup#1\@@endlink}%
\providecommand \@sanitize@url [0]{\catcode `\\12\catcode `\$12\catcode
  `\&12\catcode `\#12\catcode `\^12\catcode `\_12\catcode `\%12\relax}%
\providecommand \@@startlink[1]{}%
\providecommand \@@endlink[0]{}%
\providecommand \url  [0]{\begingroup\@sanitize@url \@url }%
\providecommand \@url [1]{\endgroup\@href {#1}{\urlprefix }}%
\providecommand \urlprefix  [0]{URL }%
\providecommand \Eprint [0]{\href }%
\providecommand \doibase [0]{https://doi.org/}%
\providecommand \selectlanguage [0]{\@gobble}%
\providecommand \bibinfo  [0]{\@secondoftwo}%
\providecommand \bibfield  [0]{\@secondoftwo}%
\providecommand \translation [1]{[#1]}%
\providecommand \BibitemOpen [0]{}%
\providecommand \bibitemStop [0]{}%
\providecommand \bibitemNoStop [0]{.\EOS\space}%
\providecommand \EOS [0]{\spacefactor3000\relax}%
\providecommand \BibitemShut  [1]{\csname bibitem#1\endcsname}%
\let\auto@bib@innerbib\@empty
\bibitem [{\citenamefont {Julien}(2010)}]{julien2010}%
  \BibitemOpen
  \bibfield  {author} {\bibinfo {author} {\bibfnamefont {A.}~\bibnamefont
  {Julien}},\ }\bibfield  {title} {\bibinfo {title} {Complexity and
  {C}ohomology for {C}ut-and-{P}rojection {T}ilings},\ }\href@noop {}
  {\bibfield  {journal} {\bibinfo  {journal} {Ergodic Theory and Dynamical
  Systems}\ }\textbf {\bibinfo {volume} {30}},\ \bibinfo {pages} {489}
  (\bibinfo {year} {2010})}\BibitemShut {NoStop}%
\bibitem [{\citenamefont {Saad}(2011)}]{saad2011}%
  \BibitemOpen
  \bibfield  {author} {\bibinfo {author} {\bibfnamefont {Y.}~\bibnamefont
  {Saad}},\ }\href@noop {} {\emph {\bibinfo {title} {Numerical {M}ethods for
  {L}arge {E}igenvalue {P}roblems: {R}evised {E}dition}}}\ (\bibinfo
  {publisher} {SIAM},\ \bibinfo {year} {2011})\BibitemShut {NoStop}%
\bibitem [{\citenamefont {Li}(2005)}]{superlu}%
  \BibitemOpen
  \bibfield  {author} {\bibinfo {author} {\bibfnamefont {X.~S.}\ \bibnamefont
  {Li}},\ }\bibfield  {title} {\bibinfo {title} {An {O}verview of {SuperLU}:
  {A}lgorithms, {I}mplementation, and {U}ser {I}nterface},\ }\href@noop {}
  {\bibfield  {journal} {\bibinfo  {journal} {ACM T.\ Math.\ Software}\
  }\textbf {\bibinfo {volume} {31}},\ \bibinfo {pages} {302} (\bibinfo {year}
  {2005})}\BibitemShut {NoStop}%
\bibitem [{\citenamefont {Kohmoto}\ \emph {et~al.}(1983)\citenamefont
  {Kohmoto}, \citenamefont {Kadanoff},\ and\ \citenamefont
  {Tang}}]{kohomoto1983}%
  \BibitemOpen
  \bibfield  {author} {\bibinfo {author} {\bibfnamefont {M.}~\bibnamefont
  {Kohmoto}}, \bibinfo {author} {\bibfnamefont {L.~P.}\ \bibnamefont
  {Kadanoff}},\ and\ \bibinfo {author} {\bibfnamefont {C.}~\bibnamefont
  {Tang}},\ }\bibfield  {title} {\bibinfo {title} {Localization {P}roblem in
  {O}ne {D}imension: {M}apping and {E}scape},\ }\href
  {https://doi.org/10.1103/PhysRevLett.50.1870} {\bibfield  {journal} {\bibinfo
   {journal} {Phys.\ Rev.\ Lett.}\ }\textbf {\bibinfo {volume} {50}},\ \bibinfo
  {pages} {1870} (\bibinfo {year} {1983})}\BibitemShut {NoStop}%
\bibitem [{\citenamefont {Ostlund}\ \emph {et~al.}(1983)\citenamefont
  {Ostlund}, \citenamefont {Pandit}, \citenamefont {Rand}, \citenamefont
  {Schellnhuber},\ and\ \citenamefont {Siggia}}]{ostlund1983}%
  \BibitemOpen
  \bibfield  {author} {\bibinfo {author} {\bibfnamefont {S.}~\bibnamefont
  {Ostlund}}, \bibinfo {author} {\bibfnamefont {R.}~\bibnamefont {Pandit}},
  \bibinfo {author} {\bibfnamefont {D.}~\bibnamefont {Rand}}, \bibinfo {author}
  {\bibfnamefont {H.~J.}\ \bibnamefont {Schellnhuber}},\ and\ \bibinfo {author}
  {\bibfnamefont {E.~D.}\ \bibnamefont {Siggia}},\ }\bibfield  {title}
  {\bibinfo {title} {{O}ne-{D}imensional {S}chr\"odinger {E}quation with an
  {A}lmost {P}eriodic {P}otential},\ }\href
  {https://doi.org/10.1103/PhysRevLett.50.1873} {\bibfield  {journal} {\bibinfo
   {journal} {Phys.\ Rev.\ Lett.}\ }\textbf {\bibinfo {volume} {50}},\ \bibinfo
  {pages} {1873} (\bibinfo {year} {1983})}\BibitemShut {NoStop}%
\bibitem [{\citenamefont {Jagannathan}(2021)}]{jaga2021}%
  \BibitemOpen
  \bibfield  {author} {\bibinfo {author} {\bibfnamefont {A.}~\bibnamefont
  {Jagannathan}},\ }\bibfield  {title} {\bibinfo {title} {The {F}ibonacci
  {Q}uasicrystal: {C}ase {S}tudy of {H}idden {D}imensions and
  {M}ultifractality},\ }\href@noop {} {\bibfield  {journal} {\bibinfo
  {journal} {Rev.\ Mod.\ Phys.}\ }\textbf {\bibinfo {volume} {93}},\ \bibinfo
  {pages} {045001} (\bibinfo {year} {2021})}\BibitemShut {NoStop}%
\bibitem [{\citenamefont {Even-Dar~Mandel}\ and\ \citenamefont
  {Lifshitz}(2008)}]{even-dar_mandel_electronic_2008}%
  \BibitemOpen
  \bibfield  {author} {\bibinfo {author} {\bibfnamefont {S.}~\bibnamefont
  {Even-Dar~Mandel}}\ and\ \bibinfo {author} {\bibfnamefont {R.}~\bibnamefont
  {Lifshitz}},\ }\bibfield  {title} {\bibinfo {title} {{Electronic Energy
  Spectra of Square and Cubic Fibonacci Quasicrystals}},\ }\href
  {https://doi.org/10.1080/14786430802070805} {\bibfield  {journal} {\bibinfo
  {journal} {Philos.\ Mag.}\ }\textbf {\bibinfo {volume} {88}},\ \bibinfo
  {pages} {2261} (\bibinfo {year} {2008})}\BibitemShut {NoStop}%
\bibitem [{\citenamefont {Damanik}\ \emph {et~al.}(2015)\citenamefont
  {Damanik}, \citenamefont {Embree},\ and\ \citenamefont
  {Gorodetski}}]{damanik2014}%
  \BibitemOpen
  \bibfield  {author} {\bibinfo {author} {\bibfnamefont {D.}~\bibnamefont
  {Damanik}}, \bibinfo {author} {\bibfnamefont {M.}~\bibnamefont {Embree}},\
  and\ \bibinfo {author} {\bibfnamefont {A.}~\bibnamefont {Gorodetski}},\
  }\bibfield  {title} {\bibinfo {title} {Spectral {P}roperties of
  {S}chr{\"o}dinger {O}perators arising in the {S}tudy of {Q}uasicrystals},\
  }in\ \href@noop {} {\emph {\bibinfo {booktitle} {Mathematics of aperiodic
  order}}}\ (\bibinfo  {publisher} {Springer},\ \bibinfo {year} {2015})\ pp.\
  \bibinfo {pages} {307--370}\BibitemShut {NoStop}%
\bibitem [{\citenamefont {Casdagli}(1986)}]{casdagli1986}%
  \BibitemOpen
  \bibfield  {author} {\bibinfo {author} {\bibfnamefont {M.}~\bibnamefont
  {Casdagli}},\ }\bibfield  {title} {\bibinfo {title} {{Symbolic Dynamics for
  the Renormalization Map of a Quasiperiodic {S}chr{\"o}dinger Equation}},\
  }\href@noop {} {\bibfield  {journal} {\bibinfo  {journal} {Commun.\ Math.\
  Phys.}\ }\textbf {\bibinfo {volume} {107}},\ \bibinfo {pages} {295} (\bibinfo
  {year} {1986})}\BibitemShut {NoStop}%
\bibitem [{\citenamefont {S{\"u}t{\H{o}}}(1987)}]{suto1987}%
  \BibitemOpen
  \bibfield  {author} {\bibinfo {author} {\bibfnamefont {A.}~\bibnamefont
  {S{\"u}t{\H{o}}}},\ }\bibfield  {title} {\bibinfo {title} {The {S}pectrum of
  a {Q}uasiperiodic {S}chr{\"o}dinger {O}perator},\ }\href@noop {} {\bibfield
  {journal} {\bibinfo  {journal} {Commun.\ Math.\ Phys.}\ }\textbf {\bibinfo
  {volume} {111}},\ \bibinfo {pages} {409} (\bibinfo {year}
  {1987})}\BibitemShut {NoStop}%
\bibitem [{\citenamefont {Colbrook}\ \emph {et~al.}(2019)\citenamefont
  {Colbrook}, \citenamefont {Roman},\ and\ \citenamefont
  {Hansen}}]{colbrook_2019}%
  \BibitemOpen
  \bibfield  {author} {\bibinfo {author} {\bibfnamefont {M.~J.}\ \bibnamefont
  {Colbrook}}, \bibinfo {author} {\bibfnamefont {B.}~\bibnamefont {Roman}},\
  and\ \bibinfo {author} {\bibfnamefont {A.~C.}\ \bibnamefont {Hansen}},\
  }\bibfield  {title} {\bibinfo {title} {{How to Compute Spectra with Error
  Control}},\ }\href {https://doi.org/10.1103/PhysRevLett.122.250201}
  {\bibfield  {journal} {\bibinfo  {journal} {Phys.\ Rev.\ Lett.}\ }\textbf
  {\bibinfo {volume} {122}},\ \bibinfo {pages} {250201} (\bibinfo {year}
  {2019})}\BibitemShut {NoStop}%
\bibitem [{\citenamefont {Hansen}(2011)}]{hansen2011}%
  \BibitemOpen
  \bibfield  {author} {\bibinfo {author} {\bibfnamefont {A.}~\bibnamefont
  {Hansen}},\ }\bibfield  {title} {\bibinfo {title} {On the solvability
  complexity index, the n-pseudospectrum and approximations of spectra of
  operators},\ }\href@noop {} {\bibfield  {journal} {\bibinfo  {journal} {J.\
  Am.\ Math.\ Soc.}\ }\textbf {\bibinfo {volume} {24}},\ \bibinfo {pages} {81}
  (\bibinfo {year} {2011})}\BibitemShut {NoStop}%
\bibitem [{\citenamefont {Ben-Artzi}\ \emph {et~al.}(2015)\citenamefont
  {Ben-Artzi}, \citenamefont {Hansen}, \citenamefont {Nevanlinna},\ and\
  \citenamefont {Seidel}}]{benartzi2015}%
  \BibitemOpen
  \bibfield  {author} {\bibinfo {author} {\bibfnamefont {J.}~\bibnamefont
  {Ben-Artzi}}, \bibinfo {author} {\bibfnamefont {A.~C.}\ \bibnamefont
  {Hansen}}, \bibinfo {author} {\bibfnamefont {O.}~\bibnamefont {Nevanlinna}},\
  and\ \bibinfo {author} {\bibfnamefont {M.}~\bibnamefont {Seidel}},\
  }\bibfield  {title} {\bibinfo {title} {Can {E}verything be {C}omputed? --
  {O}n the {S}olvability {C}omplexity {I}ndex and {T}owers of {A}lgorithms},\
  }\href@noop {} {\bibfield  {journal} {\bibinfo  {journal} {arXiv preprint
  arXiv:1508.03280}\ } (\bibinfo {year} {2015})}\BibitemShut {NoStop}%
\bibitem [{\citenamefont {Colbrook}\ and\ \citenamefont
  {Hansen}(2019)}]{colbrook2019foundations}%
  \BibitemOpen
  \bibfield  {author} {\bibinfo {author} {\bibfnamefont {M.~J.}\ \bibnamefont
  {Colbrook}}\ and\ \bibinfo {author} {\bibfnamefont {A.~C.}\ \bibnamefont
  {Hansen}},\ }\bibfield  {title} {\bibinfo {title} {The {F}oundations of
  {S}pectral {C}omputations via the {S}olvability {C}omplexity {I}ndex
  {H}ierarchy: {P}art {I}},\ }\href@noop {} {\bibfield  {journal} {\bibinfo
  {journal} {arXiv preprint arXiv:1908.09592}\ } (\bibinfo {year}
  {2019})}\BibitemShut {NoStop}%
\bibitem [{\citenamefont {Blum}\ \emph {et~al.}(1998)\citenamefont {Blum},
  \citenamefont {Cucker}, \citenamefont {Shub},\ and\ \citenamefont
  {Smale}}]{blum1998}%
  \BibitemOpen
  \bibfield  {author} {\bibinfo {author} {\bibfnamefont {L.}~\bibnamefont
  {Blum}}, \bibinfo {author} {\bibfnamefont {F.}~\bibnamefont {Cucker}},
  \bibinfo {author} {\bibfnamefont {M.}~\bibnamefont {Shub}},\ and\ \bibinfo
  {author} {\bibfnamefont {S.}~\bibnamefont {Smale}},\ }\href@noop {} {\emph
  {\bibinfo {title} {Complexity and {R}eal {C}omputation}}}\ (\bibinfo
  {publisher} {Springer},\ \bibinfo {year} {1998})\BibitemShut {NoStop}%
\bibitem [{\citenamefont {Combes}\ and\ \citenamefont
  {Thomas}(1973)}]{combes_thomas}%
  \BibitemOpen
  \bibfield  {author} {\bibinfo {author} {\bibfnamefont {J.-M.}\ \bibnamefont
  {Combes}}\ and\ \bibinfo {author} {\bibfnamefont {L.}~\bibnamefont
  {Thomas}},\ }\bibfield  {title} {\bibinfo {title} {Asymptotic {B}ehaviour of
  {E}igenfunctions for {M}ultiparticle {S}chr{\"o}dinger {O}perators},\
  }\href@noop {} {\bibfield  {journal} {\bibinfo  {journal} {Communications in
  Mathematical Physics}\ }\textbf {\bibinfo {volume} {34}},\ \bibinfo {pages}
  {251} (\bibinfo {year} {1973})}\BibitemShut {NoStop}%
\end{thebibliography}%

\end{document}


\selectlanguage{english}

\onecolumngrid
\newpage
\setcounter{figure}{0}    
\renewcommand{\thefigure}{S\arabic{figure}}
\patchcmd{\section}{\centering}{\raggedright}{}{}
\patchcmd{\subsection}{\centering}{\raggedright}{}{}
\raggedbottom
\setlength{\parskip}{0pt}
\setlength{\textfloatsep}{10pt plus 0pt minus 0pt}

\noindent {\LARGE Supplemental Material}
\vspace*{0.2cm}\\
\noindent {to ``Finding spectral gaps in quasicrystals.''}\\
\noindent {Paul Hege, Massimo Moscolari, Stefan Teufel}

\makeatletter
\def\l@subsection#1#2{}
\def\l@subsubsection#1#2{}
\makeatother

\tableofcontents

\section{1. Extended proofs}
\label{sec-periodic}

\noindent
In this section we prove the theorem and the proposition from the main text. For the convenience of the reader, we briefly recall our framework.

The Hamiltonian $H$ is a bounded hermitian operator on  $\ell^2(\Gamma;\mathcal{H})$, where  $\Gamma \subset \mathbb{R}^d$ is a countable set and $\mathcal{H}$ a separable Hilbert space. We assume that there is a maximum distance $r > 0$ such that any point in $\mathbb{R}^d$ is at most at distance~$r$ away from a point in $\Gamma$. Recall that distances in $\mathbb{R}^d$ are measured in the norm  $\|x\|_\infty := \max\{ |x_1|,\ldots, |x_d|\}$. The open cube with side-length $2r$ around $x\in\mathbb{R}^d$ is   
\[
B_r(x) := \{ y\in \mathbb{R}^d\,|\, \|x-y\|_\infty <r\}
\]
and its closure $\overline{B_r}(x) := \{ y\in \mathbb{R}^d\,|\, \|x-y\|_\infty \leq r\}$.
For any set $A\subset \mathbb{R}^d$ and wave function $\psi\in \ell^2(\Gamma;\mathcal{H})$, we denote by  $\norm{\psi}_A^2 := \sum_{x\in \Gamma \cap A} \|\psi(x)\|_{\mathcal{H}}^2$ its $\ell^2$-mass  in the region~$A$.

The following definition makes precise the idea that a Hamiltonian is gapped in the bulk at a certain energy $\lambda$ and a certain length scale $L$,  if all $\epsilon$-quasimodes at that energy that are supported in a region of that scale \eqref{eq-prop} are disproportionately supported near the edge of the region \eqref{sup:eq-prop-conclusio}.

\begin{definition}
	\label{sup:def-eps-bulk-gapped}
	Let $\epsilon,L>0$ and $\lambda\in\mathbb{R}$. We say that a Hamiltonian $H$ is \textbf{locally $\epsilon$-bulk-gapped at energy $\lambda$ and scale $L$}, if there exist constants  $N \in \mathbb{N}$, $N\geq 2$ and $C<\frac{1}{N^d} $, such that for any $x\in\Gamma$ and for any $\psi \in \ell^2(\Gamma;\mathcal{H})$, we have the following implication: Whenever
	\begin{align}
	\norm{ (H - \lambda) \psi}_ {B_L(x)} \leq \epsilon \norm{\psi}_{B_{L}(x)}\,,
	\label{eq-prop}
	\end{align}
then for $l := \frac{L + r}{N} + r$ it holds that
\begin{align}
\label{sup:eq-prop-conclusio}
\norm{\psi}_{\overline{B_l}(x)}^2 \leq C \norm{\psi}_{B_L(x)}^2\,.
\end{align}
\end{definition}

In the following we call any   $\psi \in L^2(\mathbb{R}^d)$ that satisfies the inequality \eqref{eq-prop}   an \textbf{ $L$-local $\epsilon$-quasimode at $x$}. In these terms, Definition~\ref{sup:def-eps-bulk-gapped} says that every $L$-local $\epsilon$-quasimode at $x$ has most of its mass in $B_L(x)$ concentrated outside of $\overline{B_l}(x)$. This terminology is useful in the proof of the main theorem.

\begin{theorem}\label{main-theorem-recall}
If a Hamiltonian $H$ on $\ell^2(\Gamma; \mathcal{H})$ is locally $\epsilon$-bulk-gapped at energy $\lambda$ on some scale $L>0$, then the interval  $(\lambda - \epsilon, \lambda + \epsilon)$ is a gap in the spectrum of $H$, i.e.\
\begin{align*}
\sigma(H) \cap (\lambda - \epsilon, \lambda + \epsilon) = \varnothing.
\end{align*}
\end{theorem}

\begin{proof}
Assume that $H$ is $\epsilon$-bulk-gapped at energy $\lambda$ and scale $L>0$. Let $  l, N$ and $C$ be as in Definition~\ref{sup:def-eps-bulk-gapped}. 
We generalise the proof strategy presented in the main text.
For every $q \in \{1, \dots, N\}^d$ let
\begin{align}
\widetilde Z_q := 2(L + r)(\mathbb{Z}^{d} + q/N)\label{eq-def-tilde-Zq}
\end{align}
and   
\begin{align}
f_q: \widetilde  Z_q \to \Gamma\,,\; \widetilde x\mapsto f_q(\widetilde x) := \text{the point in $\Gamma$ closest to $\tilde x$}\,.
\end{align}
In the case where several   points  in $\Gamma$ minimize the distance to $\widetilde x$, we   arbitrarily  pick one of them in the definition of $f_q(\widetilde x)$. 
By the assumption on $\Gamma$ we have that $\|\widetilde x- f_q(\widetilde x)\|_\infty \leq  r$ for all $\widetilde x\in\widetilde Z_q$. Since two distinct points in $\widetilde Z_q$ have at least distance $2(L + r)$, this implies that the map $f_q$ is one-to-one. Hence, as a map onto its image
\begin{align}
Z_q :=  f_q(\widetilde Z_q) \subset \Gamma\,,
\label{eq-def-Zq}
\end{align}
the map $f_q: \widetilde  Z_q \to Z_q$ is a bijection and we think of $Z_q$ as a deformation of~$\widetilde Z_q$.

Is is now straightforward to see that the large boxes $B_L(x)$ are still mutually disjoint when $x$ varies in one of the deformed sublattices $Z_q$ and that the small boxes  $\overline{B_l}(x)$ still cover all of $\mathbb{R}^d$  if $x$ varies in the union $\cup_q Z_q$:\\[2mm]
\noindent\textit{Disjointness:} For any $q \in \{1, \dots, N\}^d$ and two different points $x, y \in Z_q$ 
\begin{align}
B_L(x) \cap B_L(y) = \varnothing\,.
\label{eq-intersections-disjoint}
\end{align}
\begin{proof}Let $\widetilde x := f_q^{-1}(x)$ and $\widetilde y:= f_q^{-1}(y)$. Then $\widetilde x$ and $\widetilde y$ are two distinct points in the square lattice $\widetilde Z_q$ and thus their distance is at least $2(L + r)$.  Using the inverse triangle inequality we conclude that
\[
\|x-y\|_\infty = \|\tilde x -\tilde y + x-\tilde x +  \tilde y - y\|_\infty \geq  \|\tilde x -\tilde y\|_\infty - \|x-\tilde x\|_\infty -\|  \tilde y - y\|_\infty  \geq  2(L+r) - 2r=2L\,.
\]
Thus, the boxes $B_L(x)$ and  $B_L(y)$ do not overlap.
\end{proof}

\noindent \textit{Covering:}
\begin{align}
\bigcup_{q \in \{1, \dots, N\}^d} \bigcup_{x \in Z_q} \overline{B_l}(x) = \mathbb{R}^d.
\label{eq-double-union}
\end{align}
\begin{proof}
Note that the union
\begin{align*}
\widetilde Z  := \bigcup_{q \in \{1, \dots, N\}^d} \widetilde Z_q
\end{align*}
is a square lattice with side length $2(L + r) / N$. Thus, for every $p \in \mathbb{R}^d$, there exists a $q  \in \{1, \dots, N\}^d$ and  $\widetilde x \in \widetilde Z_q$ such that $\|p- \widetilde x\|_\infty \leq (L + r) / N$.
The triangle inequality now implies 
\begin{align}
\| p- f_q(\widetilde x) \|_\infty \leq \|p-  \widetilde x  \|_\infty + \|\widetilde x- f_q(\widetilde x) \|_\infty \leq (L + r) / N + r=l\, 
\end{align}
and therefore  $p \in \overline{B_l}(x)$ for $x := f_q(\widetilde x) \in Z_q$.
\end{proof}

 These two properties allow us to prove by contradiction that the spectrum $\sigma(H)$ does not contain the interval $(\lambda - \epsilon, \lambda + \epsilon)$.  
  
  Suppose there exists   $\nu \in (\lambda - \epsilon, \lambda + \epsilon)\cap\sigma(H) $. Then, according to Weyl's criterion,  there exist arbitrarily precise quasimodes for the energy $\nu$. More precisely, for any $\delta>0$  there exists $\psi \in \ell^2(\Gamma, \mathcal{H})$   such that  
\begin{align} \label{eq-quasimode}
\norm{\psi} = 1 \quad\mbox{and}\quad 
\| (H - \nu) \psi \|  < \delta .
\end{align}
We now fix $\delta > 0$ such that 
\begin{subequations}
\begin{gather}
\delta < (\epsilon - |\lambda - \nu|)^2 \label{eq-delta-bound-a}\\
\delta < N^{-d} - C \label{eq-delta-bound-b}
\end{gather}
\end{subequations}
and  choose a corresponding $\psi \in \ell^2(\Gamma, \mathcal{H})$ that satisfies \eqref{eq-quasimode}.

Notice that if $\psi$ is a $\delta$-quasimode on $\mathbb{R}^d$, this does not imply that $\psi$ is an $L$-local $\delta$-quasimode for all $x \in \mathbb{R}^d$. However, we have that the total mass of $\psi$ on those squares $B_L(x)$ for which $\psi$ is not an $L$-local $\delta$-quasimode is small (it is of order $\delta$). To see this, we split each $Z_q$ into two subsets 
\begin{align*}
Z_q^+ := \big\{  x \in  Z_q  \, |\, \norm{(H - \nu) \psi}_{B_L(x)} \leq \sqrt{\delta} \norm{\psi}_{B_L(x)} \big\}\,,\\
Z_q^- := \big\{  x \in Z_q  \, |\, \norm{(H - \nu) \psi}_{B_L(x)} > \sqrt{\delta} \norm{\psi}_{B_L(x)}\big\}\,.
\end{align*}
Around the points $x \in Z_q^+$, $\psi$ is an $L$-local $\epsilon$-quasimode, since
\begin{align} \label{eq:LEpsQuasimode}
 \norm{(H - \lambda)\psi }_{B_L(x)}
 &\leq \|(H - \nu) \psi \|_{B_L(x)} \hspace{-3pt}+ |\lambda - \nu| \norm{\psi}_{B_L(x)} \nonumber
 \\ &\leq (\sqrt{\delta} + |\lambda - \nu|) \norm{\psi}_{B_L(x)}\nonumber
 \\ & < \epsilon \norm{\psi }_{B_L(x)}\,.
\end{align}
For the last step of the above, we rewrite \eqref{eq-delta-bound-a} as
\begin{align*}
\delta < (\epsilon - | \lambda - \nu|)^2 \quad \Rightarrow \quad \sqrt{\delta} < \epsilon - |\lambda - \nu |
\quad \Rightarrow \quad 
\sqrt{\delta} + | \lambda - \nu | < \epsilon\,.
\end{align*}
Since  $H$ is $\epsilon$-bulk-gapped on the scale $L$, \eqref{eq:LEpsQuasimode} implies that for every $x \in Z^+_q$ the inequality \eqref{sup:eq-prop-conclusio} holds true. Summing over all $x \in Z_q^+$, we get
\begin{align}
\sum_{x \in Z_q^+} \norm{\psi}_{\overline{B_l}(x)}^2 \leq C \sum_{x \in Z_q^+} \norm{\psi}_{B_L(x)}^2\,.
\label{sum-Zplus}
\end{align}
By the disjointness condition shown before, we have that $B_L(x)$ and $B_L(y)$ are disjoint for different $x$ and $y$ in $Z^+_q$.  
Since for disjoint sets $A$ and $B$, we have
$
\norm{\psi}_A^2 + \norm{\psi}_B^2 = \norm{\psi}_{A \cup B}^2
$, we obtain
\begin{equation}
\label{eq:aux1}
\sum_{x \in Z_q^+} \norm{\psi}_{B_L(x)}^2 = \norm{\psi}_U^2 \leq \norm{\psi}^2 = 1\quad\text{for } U := \bigcup_{x \in Z_q^+} B_L(x) .
\end{equation}
Combining \eqref{eq:aux1} with \eqref{sum-Zplus}, we get
\begin{align}
\sum_{x \in Z_q^+} \norm{\psi}_{\overline{B_l}(x)}^2 \leq C.
\label{eq-small-Zplus-sum}
\end{align}
Let now $x \in Z_q^-$, then  by definition of  $Z_q^-$ we have \begin{equation} 
\label{eq:aux2}
\norm{\psi}^2_{B_L(x)} < \frac{1}{\delta}\norm{(H - \nu) \psi}^2_{B_L(x)}  \, .
\end{equation}
By taking the sum over all $x \in Z_q^-$ we get
\begin{align*}
\sum_{x \in Z_q^-} \norm{\psi}^2_{\overline{B_l}(x)} \leq \sum_{x \in Z_q^-}  \norm{\psi}^2_{B_L(x)} < \frac1\delta \sum_{x \in Z_q^-} \norm{(H - \nu) \psi}_{B_L(x)}^2 \leq \frac1\delta \norm{(H - \nu) \psi}^2 \,,
\end{align*}
  where we used the disjointness of the $B_L(z)$ for $z \in Z_q$ in the last step. Since $\psi$ is a $\delta$-quasimode, meaning $\norm{(H - \nu) \psi}^2 < \delta^2$, we get
\begin{align}
\sum_{x \in Z_q^-} \norm{\psi}^2_{\overline{B_l}(x)} < \delta \, .
\label{eq-small-Zminus-sum}
\end{align}
 Combining  \eqref{eq-small-Zplus-sum} and \eqref{eq-small-Zminus-sum} and using the covering property \eqref{eq-double-union}, we finally obtain
\begin{align*}
\norm{\psi}^2 &\leq \sum_{q \in \{ 1, \dots , N\}^d} \sum_{x \in Z_q} \norm{\psi}^2_{\overline{B_l}(x)} = \sum_{q \in \{ 1, \dots , N\}^d} \left( \sum_{x \in Z_q^+} \norm{\psi}^2_{\overline{B_l}(x)} + \sum_{x \in Z_q^-} \norm{\psi}^2_{\overline{B_l}(x)}  \right)\\
&\leq  \sum_{q \in \{ 1, \dots , N\}^d} (C + \delta) \;= \;N^d (C + \delta)  \; <\;1\,,
\end{align*}
where in the last inequality we used the hypothesis \eqref{eq-delta-bound-b} on $\delta$. Since $\norm{\psi}^2 < 1$ contradicts the  normalisation of $\psi$ assumed in \eqref{eq-quasimode}, such  a $\delta$-quasimode   cannot exist,  and $\nu$ is not in the spectrum of $H$.
\end{proof}

 \begin{remark} Notice that neither Definition \ref{sup:def-eps-bulk-gapped} nor Proposition \ref{sup:eq-prop-conclusio} require the boundedness of the Hamiltonian and can be easily generalized to unbounded operators taking into account only vectors that belong to the domain of the Hamiltonian.
\end{remark}

Next we prove the  numerically verifiable criterion for the $\epsilon$-bulk-gapped property. Recall that for any set $A\subset \mathbb{R}^d$ we denote by $\mathbf{1}_A$ the characteristic function of $A$ and we use the shorthand notation
$H_A :=\mathbf{1}_A H \mathbf{1}_A$.

\begin{proposition}
\label{proposition-recall}
In addition to the assumption of Theorem~\ref{main-theorem-recall},
 let $H$ have finite range, i.e.\ there exists a maximal hopping distance $m$ such that
\begin{align*}
H_{xy} = 0\quad\text{for all\; $x,y\in\Gamma$ \;with\; $\norm{x-y}_\infty > m$}.
\end{align*} 
Let $L > 0$, $N \in \mathbb{N}$, $N\geq2$, $l := \frac{L + r}{N} + r$, $\lambda\in\mathbb{R}$. 
Assume that for every $x \in \Gamma$   there exists   a set  $A(x)\subset \Gamma$ such that 
$\overline{B_l}(x)\,\subseteq\, A(x)\subseteq B_{L-m}(x)$,  $\lambda \notin \sigma(H_{A(x)})$, and
\begin{align*}
D(x) =\norm{\mathbf{1}_{\overline{B_l}(x)} (H_{A(x)} - \lambda)^{-1}  \mathbf{1}_{A(x)} H \mathbf{1}_{B_L(x) \backslash A(x)}}_\mathrm{op}
\end{align*}
satisfies $D(x)<N^{-d/2}$.
Then $H$ is $\epsilon$-bulk-gapped at energy $\lambda$ and scale $L$ for any~$\epsilon>0$ with 
\begin{align}
\epsilon < \inf_{x\in\Gamma} \frac{N^{-d/2} - D(x)}{\norm{\mathbf{1}_{\overline{B_l}(x)} (H_{A(x)}- \lambda)^{-1} \mathbf{1}_{A(x)}}_\mathrm{op}}\,.
\label{eq-conditions-epsilon-recall}
\end{align}
\end{proposition}

\begin{proof}
 Let $x \in \Gamma$ and suppose that for some $\psi  \in\ell^2(\Gamma;\mathcal{H}) $ property \eqref{eq-prop} holds, i.e.\ 
 \[
 u := \mathbf{1}_{B_L(x)} (H - \lambda) \psi
 \]
 satisfies 
 \begin{equation}
 \norm{u} < \epsilon \norm{\psi}_{B_L(x)}\,,\label{eq-u-bound}
\end{equation}
with $\epsilon>0 $ satisfying \eqref{eq-conditions-epsilon-recall}. We need to show that this implies \eqref{sup:eq-prop-conclusio}.
Writing
\begin{align*}
(H_{A(x)} - \lambda) \mathbf{1}_{A(x)} \psi &= \mathbf{1}_{A(x)} (H - \lambda) \mathbf{1}_{A(x)} \psi\\
&= \mathbf{1}_{A(x)} (H - \lambda) \psi - \mathbf{1}_{A(x)} (H - \lambda) \mathbf{1}_{{A(x)}^c} \psi\\
&= \mathbf{1}_{A(x)} u -\mathbf{1}_{A(x)} H \mathbf{1}_{{A(x)}^c} \psi 
\end{align*}
and multiplying this equality by  $(H_{A(x)} - \lambda)^{-1}$  gives
\begin{align*}
\mathbf{1}_{A(x)} \psi = (H_{A(x)} - \lambda)^{-1} (\mathbf{1}_{A(x)} u + \mathbf{1}_{A(x)} H \mathbf{1}_{{A(x)}^c} \psi)\,.
\end{align*}
Since we need to estimate $\norm{\psi}_{\overline{B_l}(x)}$, we can multiply by $\mathbf{1}_{B_l(x)}$ to obtain
\begin{align*}
\mathbf{1}_{B_l(x)}\psi &= \mathbf{1}_{\overline{B_l}(x)} (H_{A(x)} - \lambda)^{-1} (\mathbf{1}_{A(x)} u + \mathbf{1}_{A(x)} H \mathbf{1}_{{A(x)}^c} \psi)\,.
\end{align*}
Using the triangle inequality, we obtain
\begin{align}
\norm{\psi}_{\overline{B_l}(x)} &\leq \norm{\mathbf{1}_{\overline{B_l}(x)} (H_{A(x)} - \lambda)^{-1} \mathbf{1}_{A(x)} u} + \norm{\mathbf{1}_{\overline{B_l}(x)} (H_{A(x)} - \lambda)^{-1} \mathbf{1}_{A(x)} H \mathbf{1}_{{A(x)}^c} \psi} \,.
\label{eq-after-triangle}
\end{align}
The first term is easily bounded using  \eqref{eq-u-bound},
\begin{align*}
\norm{\mathbf{1}_{\overline{B_l}(x)} (H_{A(x)} - \lambda)^{-1} \mathbf{1}_{A(x)} u} < \norm{\mathbf{1}_{\overline{B_l}(x)} (H_{A(x)}- \lambda)^{-1} \mathbf{1}_{A(x)}}_\mathrm{op} \epsilon \norm{\psi}_{B_L(x)} =: M(x) \;\epsilon \norm{\psi}_{B_L(x)} \,.
\end{align*}
  Because of the finite range hypothesis on $H$ and the fact that ${A(x)} \subseteq B_{L-m}(x)$, we can rewrite the second term   as
$$
\mathbf{1}_{\overline{B_l}(x)} (H_{A(x)} - \lambda)^{-1} \mathbf{1}_{A(x)} H \mathbf{1}_{{A(x)}^c} \psi= \mathbf{1}_{\overline{B_l}(x)} (H_{A(x)} - \lambda)^{-1} \mathbf{1}_{A(x)} H \mathbf{1}_{B_L(x) \setminus {A(x)}} \psi.
$$
Thus, the second term in \eqref{eq-after-triangle} can be estimated using the assumption on $D(x)$ and we find that
\begin{align*}
\norm{\psi}_{\overline{B_l}(x)} < \epsilon M(x) \, \norm{\psi}_{B_L(x)} + D(x) \norm{\psi}_{B_L(x) \setminus {A(x)}} \leq (D(x) + \epsilon M(x)) \norm{\psi}_{B_L(x)}\,.
\end{align*}
Since   \eqref{eq-conditions-epsilon-recall} implies that  
\[
\sup_{x\in\Gamma} (D(x) + \epsilon M(x))^2 < \sup_{x\in\Gamma} \left(D(x) + \frac{N^{-d/2} - D(x)}{M(x)} M(x)\right)^2   = N^{-d}\,,
\]
  \eqref{sup:eq-prop-conclusio} holds for any 
 $C$ with $\sup_{x\in\Gamma} (D(x) + \epsilon M(x))^2 < C < N^{-d}$.
\end{proof}

\section{2. Description of the algorithms \label{sec-algorithms}}
\label{sm-sec:DescriptionAlg}
\noindent
In this section we explain in detail how to apply the general results of Theorem~\ref{main-theorem-recall}   to Hamiltonians modelled over the  
Ammann-Beenker tiling by enumerating all local patches on a certain scale $L$ and then checking the conditions of Proposition~\ref{proposition-recall}. The method can be easily generalised to other cut-and-project tilings. Moreover, Proposition~\ref{proposition-recall} can be applied in principle to any Hamiltonian that has finite local   complexity. The problem of enumerating all possible local restrictions of $H_{B_L(x)}$ on a scale $L$ then needs to be solved in a way specific to the structure of $\Gamma$ and $H$.

We first describe an algorithm to determine all the different local patches that can occur the Ammann-Beenker quasicrystal $\Gamma_\mathrm{AB}$ defined with the cut-and-project method, that is to determine  the set   $\mathcal{C}_L:= 
\{\mathcal{C}_L(x)\,|\, x\in\Gamma_\mathrm{AB}\}$ of all local patches 
$
\mathcal{C}_L(x) := \{ y\in \mathbb{R}^d \,|\, x+ y \in \Gamma_\mathrm{AB}\cap B_L(x)\} 
$.

Recall that the vertices of the Ammann-Beenker tiling are defined as the set
\begin{align*}
\Gamma_\mathrm{AB} = \big\{\,p(z) \;\big|\; z \in \mathbb{Z}^4,\,\kappa(z) \in R\,\big\}\,,
\end{align*}
where  $a=\frac{1}{\sqrt{2}}$, $R\subset \mathbb{R}^2$ is the regular axis-aligned octagon centered at $0$ with side length $1$, and
\begin{align*}
p = \begin{pmatrix}
1 & a & 0 & - a\\
0 & a & 1 & a
\end{pmatrix} \,, \quad
\kappa = \begin{pmatrix}
1 & -a & 0 & a\\
0 & a & -1 & a
\end{pmatrix}
\end{align*}
are the two ``projections'' as maps from $\mathbb{R}^4$ to $\mathbb{R}^2$.

The first algorithm we describe determines the set $V_L \subseteq \mathbb{Z}^4$ of ``candidate points'' defined as
\begin{align*}
V_L = \{\, v \in \mathbb{Z}^4 \;|\; p(v) \in  B_L(0) \;\text{and}\; \kappa(v) \in 2 R \,\}.
\end{align*}
As explained in the main text, this set is defined such that for every $z \in \mathbb{Z}^4$, the points of $\Gamma_\mathrm{AB} \cap B_L(p(z))$ are all of the form $p(z + v)$ for some $v \in V_L$. Thus to determine what points are part of the local patch around $p(z)$, we only need to check the points $z + v$ for $v \in V_L$.

The set  of candidate points $V_L$ only depends on $L$ and hence we only have to be compute it once at the beginning of the algorithm. To compute it, it would be possible in principle to simply check the two conditions $p(v) \in B_L(0)$ and $\kappa(v) \in 2R$ for all integer points in a four-dimensional cube around $0$ with radius $2L$, say, but this would be very inefficient as it requires checking $O(L^4)$ points. Because the condition $\kappa(v) \in R$ means that all points in $V_L$ lie close to the two-dimensional hyperplane defined by $\kappa(v) = 0$, it should only be necessary to check the conditions for $O(L^2)$~points.

To compute $V_L$ efficiently, consider the matrix
\begin{align*}
t = \begin{pmatrix}
1 & 0 \\
a & a \\
0 & 1 \\
-a & a
\end{pmatrix}\,.
\end{align*}
The matrix $t$ satisfies $\kappa t = 0$ and $p t = 2 \cdot \mathbf{1}$. Because the columns of $t$ together with the second and fourth canonical basis vectors $e_2$ and $e_4$ form a linear basis of $\mathbb{R}^4$, we can write any $v \in V_L$ as
\begin{align}
v = t\begin{pmatrix}w_1 \\ w_2\end{pmatrix} + q_1 e_2 + q_2 e_4 = \begin{pmatrix} w_1 \\ a(w_1+w_2) +q_1\\ w_2 \\ a(w_2-w_1) +q_2
\end{pmatrix}  \,. \label{eq-v-def}
\end{align}
Because $v \in \mathbb{Z}^4$, also   $w_1,w_2\in\mathbb{Z}$.
The condition $\kappa v \in 2 R$ that holds for all $v \in V_L$ can be used to bound $q_1$ and $q_2$. In fact, using $\kappa t = 0$  we conclude that 
\begin{align}
\begin{pmatrix}a & a \\ -a & a\end{pmatrix} \begin{pmatrix}q_1 \\ q_2\end{pmatrix} = \kappa(q_1 e_2 + q_2 e_4) = \kappa(v) \in 2R\,.
\label{cond-q1-q2}
\end{align}
Since $\begin{pmatrix}a & a \\ -a & a\end{pmatrix}$ and hence also its inverse are symmetries of $R$,  \eqref{cond-q1-q2} just becomes 
\begin{align*}
\begin{pmatrix}q_1 \\ q_2\end{pmatrix} \in 2R\,.
\end{align*}
In particular, this condition on $(q_1,q_2)$ is independent of $L$. This is what allows enumerating all points according to equation~\eqref{eq-v-def} with $q_1, q_2$ in this range to be $O(L^2)$.

We now compute the range of values that $w_1$ and $w_2$ can take such that  \eqref{eq-v-def}  defines  an element $v \in V_L$. Using $pt = 2\cdot\mathbf{1}$ we find
\begin{align*}
p(v) = p \left( t \begin{pmatrix}w_1\\w_2\end{pmatrix} + q_1 e_2 + q_2 e_4\right) = 2 \begin{pmatrix}w_1 \\ w_2\end{pmatrix} + \begin{pmatrix}a & -a \\ a & a\end{pmatrix} \begin{pmatrix}q_1 \\ q_2\end{pmatrix}\,.
\end{align*}
As also $\begin{pmatrix}a & -a \\ a & a\end{pmatrix}$  leaves $R$ invariant, $(q_1, q_2) \in 2R$ implies
\begin{align*}
p(v) - 2 \begin{pmatrix}w_1 \\ w_2\end{pmatrix} \in 2 R\,.
\end{align*}
Since   $v\in V_L$ requires  $pv \in B_L(0)$ and since the $\infty$-norm of $2R$ is bounded by $1 + \sqrt{2}$, we find that
\begin{align*}
2 \begin{pmatrix}w_1 \\ w_2\end{pmatrix} \in B_{L + 1 + \sqrt{2}}(0)\,.
\end{align*}
Thus $w_1$ and $w_2$ must both lie in the interval $(-L/2 - s, L/2 + s)$ with $s=\frac{1+\sqrt{2}}{2}$.

Thus, when determining whether $v\in\mathbb{Z}^4$ lies in $V_L$,  the values of $ v_1=w_1$ and $ v_3=w_2$ can be restricted to  integers in this range (line 1 in Algorithm 1 below). For $q_1$ and $q_2$, it suffices to consider all values in $2R\subset B_{2s}(0)$ such that the resulting values for $v_2$ and $v_4$ become integral (line 2 in Algorithm 1 below).

\begin{algorithm}[H]
\caption{Enumerate ``candidate set'' $V_L$
\label{algorithm-enumerate-candidates}}
\begin{algorithmic}[1]
\ForAll{integers $v_1, v_3$   in $(-L/2 - s,  L/2 + s)$}
\ForAll{integers $v_2 \in a(v_1 + v_3) + [-2s, 2s]$ and $v_4 \in a(v_1 - v_3) + [-2s, 2s]$}
\If{$p(v) \in B_L(0)$ and $\kappa(v) \in 2 R$}
\State add $v$ to $V_L$
\EndIf
\EndFor
\EndFor
\end{algorithmic}
\end{algorithm}
Next, we describe the algorithm to enumerate the local patches of the Ammann-Beenker tiling. As we have discussed in the main text, the local patch $\mathcal{C}_L(x)$ centered at $x=p(z)\in \Gamma_\mathrm{AB}$ can be characterized by which of the candidate points in $V_L$  ``become part of the tiling''. Namely, for every $v \in V_L$, we have
\begin{align}
p(z + v) \in \Gamma_\mathrm{AB} \quad\Leftrightarrow\quad \kappa(z + v) \in R \quad\Leftrightarrow\quad \kappa(z) \in R - \kappa(v)\,.
\label{eq-equivalences}
\end{align}
Defining for $v\in\mathbb{R}^4$ the following decomposition of the acceptance region,
\begin{gather*}
 P^1(v) := R \cap  (R - \kappa(v))\quad\text{and}\quad P^0(v) := R \setminus P^1(v)\,,
\end{gather*}
   \eqref{eq-equivalences} entails that  for every $v \in V_L$ it holds that 
   \begin{equation}\label{condition1}
   p(v) \in \mathcal{C}_L(x) \quad\Leftrightarrow\quad
   \kappa(z)\in P^1(v)\,.
   \end{equation}
By labelling the points in the set $V_L$ with an index $i \in \{1,\dots, |V_L|\}$, we can uniquely associate to each bit string $b=(b_1,\dots,b_{|V_L|}) \in \{0,1\}^{|V_L|}$ a local patch
\begin{align*}
\mathcal{C}_{L,b} := \{ \,p(v_i) \, | \,    b_i = 1, i \in \{1,\dots, |V_L|\}\}\,.
\end{align*}
From \eqref{condition1} we conclude that 
\begin{equation}
\mathcal{C}_L(x) = \mathcal{C}_{L,b} \quad\Leftrightarrow
\quad \kappa(z) \in P_{(b_i)} := \bigcap_{i = 1}^{|V_L|} P^{b_i}(v_i)\,.\label{eq-def-A}
\end{equation}

For many bit strings $b$ the set $P_{(b_i)}$ turns out to be empty and thus not all local patches $\mathcal{C}_{L,b}$ defined by bit strings $b$ of length $|V_L|$ correspond to actual local patches of the Ammann-Beenker tiling. In fact, the number of local patches (also referred to as the ``patch counting function'' or, in \cite{julien2010}, ``complexity'') of the Ammann-Beenker tiling is of order $O(L^2)$, cf.\  \cite{julien2010}, while the number of possible bit strings is $O(2^{|V_L|})$, with $|V_L|$ growing like $L^2$.

To enumerate those bit strings that correspond to actual local patches (i.e.,~for which the set $P_{(b_i)}$ is not empty), we can use a dynamic programming approach. To do so, we first extend the definition \eqref{eq-def-A} of $P_{(b_i)}$ to shorter bit strings $(b_i)$, namely  to $(b_i)\in \{0,1\}^m$, $1 \leq m \leq |V_L|$, by taking into account only the intersections up to the $m$-th place of the bit string.

Our algorithm then proceeds step by step by computing all nonempty $P_{(b_i)}$ for bit strings $(b_i)$ of length $1$, of length~$2$, and so on. By \eqref{eq-def-A}, we can go from a bit string of lenght $n$ to a bit string of lenght $n+1$ using the following recursion relation:
\begin{equation}
\begin{aligned}
P_{(b_i)\oplus 1} &= P_{(b_i)} \cap P^1(v_{i+1})\\
P_{(b_i)\oplus 0} &= P_{(b_i)} \cap P^0(v_{i+1})\,.
\end{aligned}
\label{P-split}
\end{equation}

Suppose we have computed a set $J_n$ of all bit strings $(b_i)$ of length $n$ and their associated sets $P_{(b_i)}$ for all $(b_i)$ where $P_{(b_i)}$ is not empty. We can then compute the set $J_{n+1}$ by computing, for every $(b_i) \in J_n$, the two intersections on the right hand sides of \eqref{P-split} and adding those of the two bit strings $(b_i) \oplus 0$ and $(b_i) \oplus 1$ for which the intersection is not zero. This procedure can be implemented algorithmically by simply maintaining a list $J$ of nonempty regions $P_{(b_i)}$ and the associated bit strings, as described in the pseudocode in the loop in line \ref{optimization-loop-to-avoid} of Algorithm \ref{algorithm-enumerate-tilings}. The following algorithm is formulated to decompose an area $R_0 \subseteq R$. While setting $R_0 = R$ will lead to an enumeration of all patches, it is sometimes advantageous to compute only the decompositions corresponding to $\kappa(z) \in R_0$ for some smaller $R_0$, as described below.

\begin{algorithm}[H]
\caption{ Enumerate the local patches $\mathcal{C}_L(x)$}
\label{algorithm-enumerate-tilings}
\begin{algorithmic}[1]
\State Initialize $J = \{(R_0, \text{` '})\}$.
\State Compute $V_L$ using Algorithm~\ref{algorithm-enumerate-candidates}
\ForAll{$v \in V_L$}
\State set $P^1(v) = R - \kappa(v)$.
\State Initialize $J_2 = \{\}$
\If{$R_0 \subseteq P^1(v)$}\label{optimization-start}
\State set $J \leftarrow \{(Q, s \oplus\text{`1'}) \,|\, (Q,s) \in J\}$
\State \textbf{continue} with next $v$
\EndIf
\If{$R_0 \cap P^1(v) = \varnothing$}
\State set $J \leftarrow \{(Q, s \oplus\text{`0'}) \,|\, (Q,s) \in J\}$
\State \textbf{continue} with next $v$
\EndIf\label{optimization-end}
\ForAll{$(Q, s) \in J$}\label{optimization-loop-to-avoid}
\If{$P^1(v) \cap Q \neq \varnothing$}
\State Add $(P \cap Q, s \oplus \text{`1'})$ to $J_{ 2}$
\EndIf
\If{$Q \backslash P^1(v) \neq \varnothing$}
\State Add $(Q \backslash P, s \oplus \text{`0'})$ to $J_{2}$
\EndIf
\EndFor
\State Update $J \leftarrow J_2$.
\EndFor
\end{algorithmic}
\end{algorithm}

Usually we decompose the entire region of acceptance $R$ according to the previous procedure. However, it can be useful to decompose only a smaller polygon $R_0 \subseteq R$. The reason for this is twofold. First, there are some symmetries of $R$ which can allow us to decompose a smaller region. For example, the tiling corresponding to $\kappa(z) = (k_1, k_2)$ is exactly the mirror image of the tiling corresponding to $(k_2, k_1)$. Therefore, considering
\begin{align*}
R_0 = R \cap \{ (x,y) \in \mathbb{R}^2 \,|\, x > 0, y > 0, y < x \}.
\end{align*}
is sufficient for enumerating all local patches up to all mirror symmetries. In that case one may replace   $R$   by $R_0$ in Algorithm~\ref{algorithm-enumerate-tilings}.

Apart from the case where one has to check only part of the acceptance region $R$ for symmetry reasons, splitting $R$ into sub-polygons can improve the performance of the algorithm. In the default case, for every candidate point $v_i \in V_L$, the intersection of the set $P^1(v_i)$ with all polygons $P_{(b_j)}$ distinguished up to that point has to be computed. The time taken for this appears to be quadratic in the final number $|\mathcal{C}_L|$ of local patches distinguished. If one only computes the decomposition of a smaller polygon $R' \subseteq R$, the intersection computation can be skipped in all cases where either $R' \subseteq P^1(v_i)$ or $R' \subseteq P^0(v_i)$, because the intersections are trivial (either empty or all of $P_{(b_j)}$). This is implemented in lines~\ref{optimization-start} to~\ref{optimization-end} of the above algorithm by replacing $R$ by $R'$. 
 In practice, we have found that splitting the region $R$ into many convex polygonal pieces $R_j'$, such that $\bigcup_j R_j' = R_0$ (we used $80$ pieces $R_j'$ for $L = 100$) greatly improved the running time of the algorithm. The fact that some tilings could occur in several of the pieces was found to play a negligible role in terms of performance. Additionally, this method allows the algorithm to be parallelized across multiple cores or nodes, since the computation for every piece $R_i'$ is independent of all others.

Having enumerated the candidate set, we now  describe how to use   the enumeration Algorithm \ref{algorithm-enumerate-tilings}  to check whether the Hamiltonian on each patch satisfies the condition in   Proposition~\ref{proposition-recall}. Specifically, we have to compute the norms
\begin{equation}
\label{eq:Dx}
D(x) =  \norm{\mathbf{1}_{\overline{B_l}(x)} (H_{A(x)} - \lambda)^{-1}  \mathbf{1}_{A(x)} H \mathbf{1}_{B_L(x) \backslash {A(x)}}} 
\end{equation}
and check that 
\begin{align}
  D(x) < N^{-d/2}
\label{eq-resolvent-property}
\end{align}
for all $x \in \Gamma$. Let us first explain in detail how we choose $A(x)$. Since the norm of $(H_{A(x)} - \lambda)^{-1}$ is expected to fall off exponentially in the distance between $\mathbf{1}_{\overline{B_l}(x)}$ and $\mathbf{1}_{A(x)} H \mathbf{1}_{B_L(x) \backslash A(x)}$, $D(x)$ is expected to be minimal for $A(x) = B_{L - m}(x)$.
It turns out that this is true in most cases. However, 
if $\lambda$ is very close to the spectrum of $H_{B_{L-m}(x)}$, which happened in our simulations only for a 
 few local patches, then  
 the norm of $(H_{B_{L-m}(x)} - \lambda)^{-1}$ and thus $D(x)$ may become too large.
 In those cases we  computed  \eqref{eq:Dx} again with a different choice of $A(x)$.
 We found that already removing one site chosen at random from the edge set $\Gamma_\mathrm{AB} \cap (B_{L - m}(x) \backslash B_{L - m -1}(x))$ from $B_{L-m}(x)$ was usually enough to perturb the spectrum of $H_{B_{L-m}(x)}$ sufficiently to remove outliers in the value of $D(x)$.

To actually compute the norm in \eqref{eq:Dx}, the most obvious method would be to invert $(H_{A(x)} - \lambda)$, compute the matrix products as written, and then compute the norm. However, computing matrix inverses is expensive and prone to numerical error, and it would not be feasible at all for matrices of the size we are considering in this paper. The method of computing $D(x)$ we   present now, by contrast, even works with sparse matrices, which makes it efficient enough that it can be employed for values of $L$ where an investigation of the spectrum of $H_{B_L(x)}$ by direct diagonalization would not be possible (although many algorithms are of course available for approximating the spectrum of sparse Hermitian operators \cite{saad2011}).

Instead of computing the matrix product as written in the right-hand side of \eqref{eq:Dx}, we compute the individual entries of the matrix. Let $M(x) := \mathbf{1}_{\overline{B_l}(x)} (H_{A(x)} - \lambda)^{-1}  \mathbf{1}_{A(x)} H \mathbf{1}_{B_L(x) \backslash A(x)}$.
Then, for $\tilde{x} \in \Gamma \cap\overline{B_l}(x)$ and $\tilde{y} \in \Gamma \cap (B_L(x) \backslash A(x))$, we compute the matrix entries 
\begin{equation*}
M(x)_{\tilde{x}\tilde{y}} = \mathbf{1}_{\{\tilde{x}\}} (H_{A(x)} - \lambda)^{-1}  \mathbf{1}_{A(x)} H \mathbf{1}_{\{\tilde{y}\}} \, .
\end{equation*}
Depending on whether there are more $\tilde{x}$'s or $\tilde{y}$'s to consider, 
we compute either the vector 
\begin{align*}
\mathbf{1}_{\{\tilde{x}\}} (H_{A(x)} - \lambda)^{-1} \mathbf{1}_{A(x)} H
\quad\mbox{or}\quad
(H_{A(x)} - \lambda)^{-1} \mathbf{1}_{A(x)} H \mathbf{1}_{\tilde{y}}.
\end{align*}
This is easier because for a given vector $\mathbf{z}$ we can compute $(H_{A(x)} - \lambda)\mathbf{z}$ using the sparse matrix $H_{A(x)}$. The efficiency of the algorithm can be greatly increased by preparing a decomposition of $(H_{A(x)} - \lambda)$ that can be used to solve $(H_{A(x)} - \lambda)^{-1} \mathbf{z}$ for different vectors $\mathbf{z}$; we use a sparse LU decomposition \cite{superlu} for our computations. We then compute the norm $D(x)$ from these matrix entries. The following algorithm describes the computation we use to check \eqref{eq-resolvent-property}.

\begin{algorithm}[H]
\caption{Check equation \eqref{eq-resolvent-property}
\label{algorithm-check-resolvent}}
\begin{algorithmic}[1]
\State Compute the sparse LU factorization of $H_{A(x)} - \lambda$
\State Initialize an empty matrix $m \times n$ matrix $T$, where $m = |B_l(x)|_\Gamma$ and $n = |B_L(x) \backslash {A(x)}|_\Gamma$
\If{$n < m$} \label{if-resolvent-algorithm}
\ForAll{$\tilde{y} \in (B_L(x) \backslash A) \cap \Gamma$}
\State Set $b = (H\delta_{\tilde{y}})|_{A(x)}$ as a vector in $l^2({A(x)} \cap \Gamma)$
\State Solve $(H_{A(x)}- \lambda) y = b$ using the LU factorization
\State Put $y$ in the column of $T$ corresponding to $\tilde{y}$
\EndFor
\Else
\State Initialize an $m \times P$ matrix $\tilde T$, where $P = |{A(x)}|_\Gamma$.
\ForAll{$\tilde{x} \in B_l(x) \cap \Gamma$}
\State Set $b = \delta_{\tilde{x}}$ as a vector in $l^2(A \cap \Gamma)$
\State Solve $(H_{A(x)} - \lambda) y = b$ using the LU factorization
\State Put $y$ in the row of $\tilde T$ corresponding to $\tilde{x}$
\EndFor
\State Set $T = \tilde T H \mathbf{1}_{B_L(x) \backslash {A(x)}}$, an $m \times n$ matrix.
\State Set $D(x) = \norm{T}_{\mathrm{op}}$ and check $D(x) < N^{-d/2}$.
\EndIf
\end{algorithmic}
\end{algorithm}

The two branches of the ``if'' statement in line~\ref{if-resolvent-algorithm} of Algorithm \ref{algorithm-check-resolvent} always compute the same number, and the condition is only an optimization that transposes the matrix product in \eqref{eq-resolvent-property} in order to reduce the number of linear systems that have to be solved.

Finally, Algorithm \ref{algorithm-final} summarizes our general strategy, using the enumeration of the local patches from Algorithm~\ref{algorithm-enumerate-tilings} and checking the condition of Propositon~\ref{proposition-recall} on each of them using Algorithm~\ref{algorithm-check-resolvent}. The loop in line~\ref{outer-loop} can be performed on different computer nodes if necessary.

\begin{algorithm}[H]
\caption{Prove a gap at energy $\lambda$ \label{algorithm-final}}
\begin{algorithmic}[1]
\State Split the region $R$ into $n$ smaller polygons $R_1, \dots, R_j, \dots, R_n$ (for efficiency)
\ForAll{polygons $R_j$}\label{outer-loop}
\State Using Algorithm~\ref{algorithm-enumerate-tilings}, decompose $R_j$ into a number of polygons $P_{(b_i)}$ corresponding to local patches in $\mathcal{C}_L$.
\State Initialize $r_{\min, j} = \infty$, the minimum gap size
\ForAll{polygons $P_{(b_i)}$ corresponding to a bit string $(b_i)$}
\State Construct the Hamiltonian $H$ on $B_L(x)$, for an $x$ with $\kappa(p^{-1}(x)) \in P_{(b_i)}$
\State Generate a finite list $A_k$ of values of $A(x)$ to try. We always set $A_1= B_{L-m}(x)$, while $A_2, A_3, \dots, A_t$ are generated by removig random points from the edge of $A_1$.
\ForAll{$A = A_1, A_2, \dots, A_t$}
\State Set $A = B_{L-m}(x)$ \label{line-set-a}
\State Use Algorithm~\ref{algorithm-check-resolvent} to check equation \eqref{eq-resolvent-property}
\If{equation \eqref{eq-resolvent-property} is fulfilled}
\State Set $r_{\min,j} \leftarrow \min(r_{\min,j}, \epsilon],$ where $\epsilon$ is the maximum allowed by equation~\eqref{eq-conditions-epsilon-recall}.
\State \textbf{continue} with next polygon $P_{(b_i)}$
\EndIf
\EndFor
\If {Equation~\eqref{eq-resolvent-property} was not fulfilled for any $A$}
\State \textbf{end computation}, no gap could be proven
\EndIf
\EndFor
\State Let $r_{\min} = \min(r_{\min, 1}, \dots, r_{\min, n})$
\State \textbf{end computation}, the infinite Hamiltonian has a gap of size $r_{\min}$ at energy $\lambda$.
\EndFor
\end{algorithmic}
\end{algorithm}

\section{3. Fibonacci Quasicrystals}

\noindent
The Fibonacci quasicrystal is a simple one-dimensional quasicrystal that was studied even prior to the discovery of physical quasicrystals \cite{kohomoto1983, ostlund1983}.  In recent years, significant attention has been devoted to the mathematical rigorous study of the spectrum of the Hamiltonian associated to the Fibonacci quasicrystals \cite{jaga2021, even-dar_mandel_electronic_2008,damanik2014}. In particular, it has been proved that the spectrum of the  Fibonacci Hamiltonian is a Cantor set \cite{casdagli1986, suto1987}. In this section, we will describe how our method can be applied to systems in one dimension using the explicit example of the Fibonacci quasicrystal, which has the advantage that many of the constructions are easier to visualize in such case. In particular, we compute upper and lower bound for the distance to the spectrum for the Fibonacci Hamiltonian which clearly show the fractal structure of its spectrum, see Figure \ref{fig-fib-both-bounds}.

\begin{figure}[ht]
	\centering 
	\includegraphics{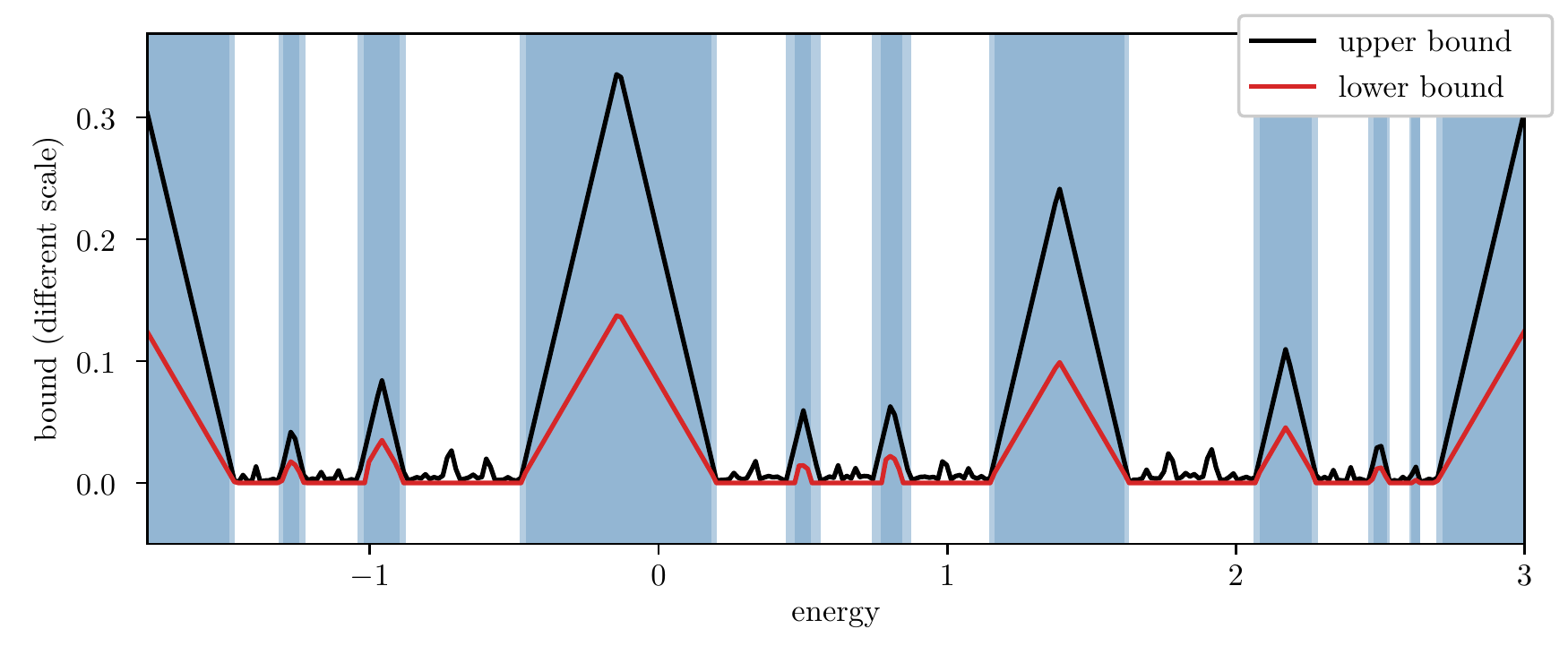}
	\caption{Upper and lower bounds for the distance to the spectrum for the Fibonacci Hamiltonian, defined by \eqref{FibonacciHamiltonian}, for $L = 500, N = 6, \alpha = 1$. The transparent blue intervals display the minimal and maximal sizes (defined as in Figure~4) of the gaps centered around the local maxima of the lower bound. For the Fibonacci quasicrystal, multiple gaps can be seen and proven to exist. As $L \to \infty$, the number of gaps will grow as the spectrum of the infinite Hamiltonian is a Cantor set.\label{fig-fib-both-bounds}}
\end{figure}

\subsection{Cut-and-project construction of the Fibonacci quasicrystal}

\noindent
As in the Ammann-Beenker case, we will define two projections, in this case from $\mathbb{R}^2 \rightarrow \mathbb{R}$, corresponding respectively to the real space and to the additional dimension:

\begin{align}
p = \begin{pmatrix}1 & \varphi\end{pmatrix} && \kappa = \begin{pmatrix}-\varphi & 1\end{pmatrix} ,
\end{align}
where $\varphi:=\frac{1+\sqrt5}{2}$ is the golden ration. Clearly the kernels of $p$ and $\kappa$ are again orthogonal.

The acceptance region in the case of the Fibonacci quasicrystal consists simply of the interval
\begin{align}
R = [0,1)\,.
\end{align}
This is the projection of the vertical interval $\{0\}\times[0,1)$ via $\kappa$.

We can then define the Fibonacci lattice as
\begin{align*}
\Gamma_{\text{Fib}} = \big\{\; p(z) \;\big|\; z \in \mathbb{Z}^2,\; \kappa(z) \in R \;\big\} \, .
\end{align*}
The condition $\kappa(z) \in R$ corresponds to the ``cutting'' step, the expression $p(z)$ is the ``projection'' step.

\begin{figure}[ht]
\includegraphics{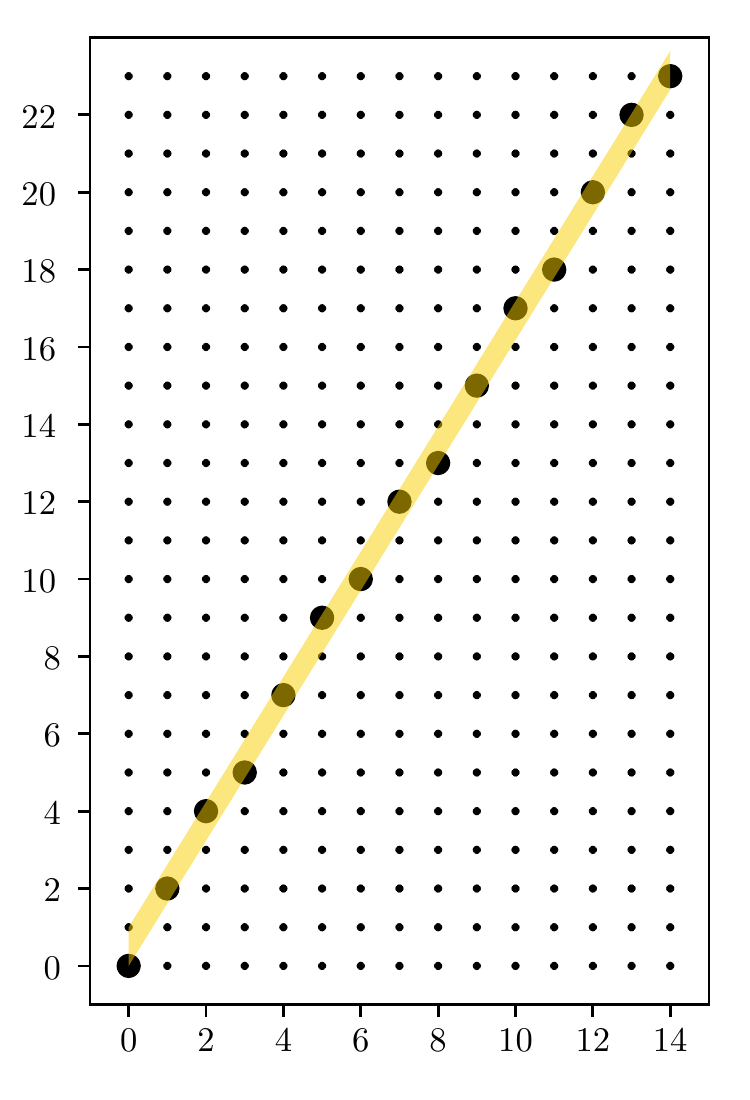}
\caption{Cut-and-project construction of the Fibonacci quasicrystal. \label{fig-three-offsets}}
\end{figure}

Figure~\ref{fig-three-offsets} contains a pictorial representation of this cut-and-project construction. The yellow shaded area in Figure~\ref{fig-three-offsets} shows points in $\mathbb{R}^2$ for which $\kappa(z) \in R$.

This definition of the Fibonacci quasicrystal is equivalent to the more common one that uses the substitution rules
\begin{align*}
\mathrm{S} \rightarrow \mathrm{L} && \mathrm{L} \rightarrow \mathrm{LS}\,.
\end{align*}
Starting from the string ``L'', this substitution rule gives a sequence of strings ``LS'', ``LSL'', ``LSLLS'', ``LSLLSLSL''$\dots$ in which the $n$-th string is a prefix of the $(n+1)$-th and a suffix of the $(n+2)$-th. This allows one to define a Fibonacci string that is infinite in both directions.

The equivalence is based on the fact that neighbouring points in the Fibonacci lattice form either ``long'' or ``short'' distances, corresponding to the letters ``L'' and ``S'' in the substitution definition. As one can see in Figure~\ref{fig-three-offsets}, successive points $z$ with $\kappa(z) \in R$ always differ by either the vector $(1, 1)$ or the vector $(1,2)$. Under the projection $p$, these vectors map to offsets
\begin{align*}
(1 + \varphi)\quad\text{and}\quad   (1 + 2 \varphi) \, .
\end{align*}
The intervals of length $(1 + 2\varphi)$ are the ``long'' intervals L and the intervals of length $(1 + \varphi)$ are the ``short'' intervals S. The quotient between the lengths is again $\varphi$.

\subsection{Enumeration of local patches}

\noindent
Let us show how to enumerate all local patches $\mathcal{C}_L$ of the Fibonacci quasicrystal. As in the case of the Ammann-Beenker tiling, for any $x \in \Gamma_{\text{Fib}}$, the \textit{local patch around $x$} is the set
\begin{align*}
\mathcal{C}_L(x) = \left\{ \,\tilde x - x \,\middle|\, \tilde x \in \Gamma_{\text{Fib}} \cap B_L(x)\,\right\}.
\end{align*}
Since every point in $x \in \Gamma_{\text{Fib}}$ has exactly one preimage $z \in \mathbb{Z}^2$ with $p(z) = x$ and $\kappa(z) \in R$, by exploiting the linearity of $p$ and $\kappa$, the set of local patch can be rewritten in the usual form as 
\begin{align*}
\mathcal{C}_L(x) = \left\{ \,p(v) \;\middle|\; v \in \mathbb{Z}^2, \kappa(v) \in R - \kappa(z), |p(v)| < L \,\right\}.
\end{align*}
This description of the local patch only depends on $\kappa(z)$, where $z$ is the integer preimage of $x$ under $p$ (which is unique by irrationality considerations). Our enumeration algorithm will decompose the interval $[0,1)$ in which $\kappa(z)$ lies into subintervals corresponding to different local patches. 

It turns out that for any given $L$, there are only finitely many points $v \in \mathbb{Z}^2$ which can fulfill the two conditions $\kappa(v) \in R - \kappa(z)$ and $|p(v)| < L$, across all $z \in \mathbb{Z}^2$ with $\kappa(z) \in R$. Indeed, for a point $v \in \mathbb{Z}^2$ to be able to satisfy the condition $\kappa(v) \in R - \kappa(z)$ for any $z$ with $\kappa(z) \in R$, we must have
\begin{align*}
\kappa(v) \in R + (- R)
\end{align*}
where $R +(- R)$ denotes the Minkowski sum. This equals the interval $R +(- R) = (-1, 1)$.  As for the Ammann-Beenker tiling, we define the set of ``candidate points'' as 
\begin{equation*}
V_L:=\{v \in \mathbb{Z}^2 \, | \, \kappa(v) \in (-1, 1) \, ; \, p(v) \in (- L, L)\, \}.
\end{equation*}

\begin{figure}[ht]
\centering \includegraphics{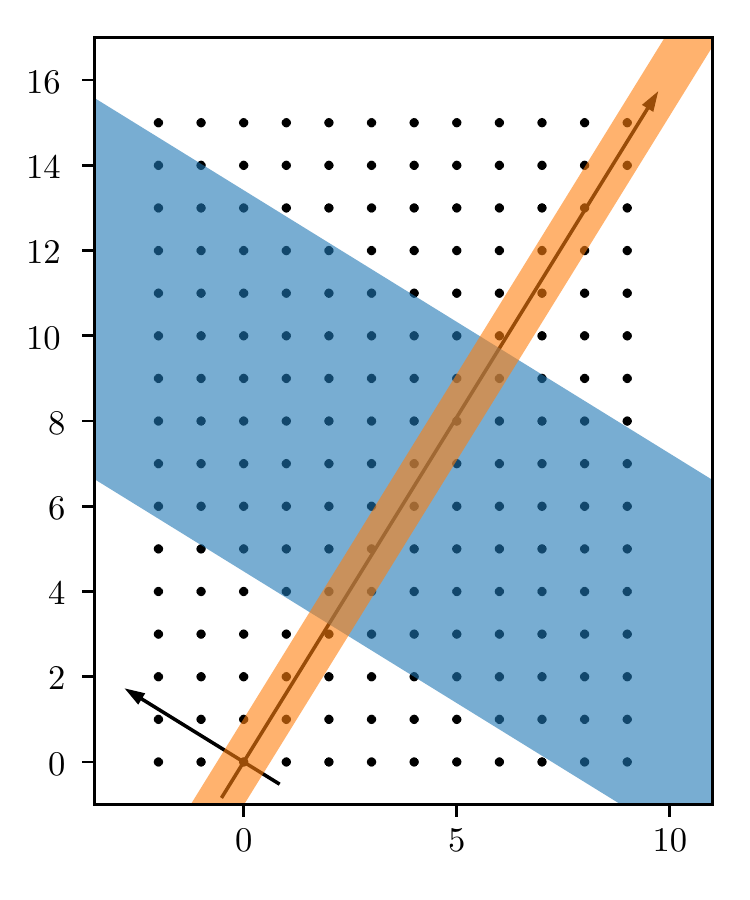}
\caption{The set of candidate points is defined by the two constraints $|p(v)| < L$ and $\kappa(v) \in (-1,1)$. In this picture, the constraint on $p(v)$ is indicated by the blue shaded area, and the restriction on $\kappa(v)$ by the orange shaded area. Because the projections have orthogonal kernel, it is clear that the intersection of both sets is compact and contains only finitely many integer points. \label{fig-candidates}}
\end{figure}

Because the two projections $p$ and $\kappa$ have orthogonal kernels, the two sets where the two conditions $\kappa(v) \in (-1, 1)$ and $p(v) \in (-L, L)$ respectively are fulfilled are two orthogonal strips, whose intersection is a rectangle, as shown in Figure~\ref{fig-candidates}. Any finite rectangle contains a finite number of points, therefore the set $V_L$ ist finite.

It is also easy to enumerate these candidate points algorithmically. To do so we describe here an alogorithm similar to Algorithm \ref{algorithm-enumerate-candidates} that we used for the Ammann-Beenker tiling. 
In this case we can choose any vector in $\mathbb{R}^2$ to be a right inverse for $p$ (up to a constant factor) and we just need to impose that such vector is in the kernel of $\kappa$. Consider the vector $t:=(1 \;\; \varphi )$, we have $ pt=(1+\varphi^2)$ and we also have that $\kappa t =0$.  Therefore, every point $v \in V_L$ can be written as 
\begin{equation*}
v=t y + q e_2
\end{equation*}
for some $y,q \in \mathbb{R}$. Moreover, $v \in \mathbb{Z}^2$ implies that $y \in \mathbb{Z}$, while $\kappa v \in (-1,1)$ implies that $q \in (-1,1)$. Thus, any candidate point is of the form $v=(v_1,v_2)$,  with $v_1 \in \mathbb{Z}$ and $v_2 \in \varphi v_1 + (-1,1)$. If we then consider also the condition $pv \in (-L,L)$, we get that $v_1 \in \frac{(-L,L) + (-1,1)}{1+\varphi^2}$. In particular, if we want to count all the candidate points $v$ such that $pv \in I$, with $I$ a given interval, we have to enumerate all the integers $v_1,v_2$ such that 
\begin{equation*}
v_1 \in \frac{I+ (-1,1)}{1+\varphi^2} \qquad v_2 \in \varphi v_1 + (-1,1) \,
\end{equation*}
which results in two loops (over $v_1$ and $v_2$) similar two the loops in Algorithm~\ref{algorithm-enumerate-candidates}.

As we have explained in Section 2 \ref{sm-sec:DescriptionAlg}, once a set of candidate points has been computed, we can categorize the values $\kappa(z)$ by which local patch we get. The local patch is completely determined by which points $\tilde{z}:=z + v \in \mathbb{Z}^2$ become part of the tiling. All points in the candidate set $V_L$ fulfill
\begin{align*}
|p(v)| < L.
\end{align*}
Therefore $|p(z + v) - p(z)| < L$. Thus, the condition $|p(\tilde z) - p(z)| < L$ is always fulfilled for all points in our candidate set.

Instead, whether the other condition is fulfilled, namely $\kappa(\tilde z) \in R$, depends on the base point $z$, or more precisely on $\kappa(z)$. We have $\kappa(\tilde z) = \kappa(z + v) = \kappa(z)  + \kappa(v) \in R$, which implies
\begin{gather}
\kappa(z) \in R-\kappa(v) .\label{eq-whichisanint}
\end{gather}
For a given candidate point $v \in V_L$, the right hand side of \eqref{eq-whichisanint} is just an interval, thus for every candidate point $v \in V_L$, we get one such interval. The local patch of scale $L$ around $p(z)$ only depends on which of these intervals the point $\kappa(z)$ lies in.

Let us enumerate the candidate points as $V_L = \{v_i \,|\, i = 1, \dots, |V_L| \}$. This corresponds to an enumeration of the corresponding intervals $I_i =
R - \kappa(v_i)$. Now, a local patch can be described by a bit string $b \in \{ 0, 1 \}^{|V_L|}$, by setting $b_i = 1$ if $\kappa(z) \in I_i$ and $b_i = 0$ otherwise.

Of course, not every bit string will occur in this way. For every bit string $b$, the portion of the acceptance region for which this patch occurs is given by
\begin{gather*}
P_{(b_i)} = \bigcap_{i = 1}^{|V_L|} P^{b_i}(v_i)
\intertext{where}
P^{1}(v_i) = R \cap I_i, \quad
P^{0}(v_i) = R \setminus P^1(v_i).
\end{gather*}
For every $k = \kappa(z) \in P_{(b_i)}$, the bit string associated to the local patch around $x=p(z)$ is exactly $P_{(b_i)}$. If a bit string $b=(b_i)$ does not actually occur, the set $P_{(b_i)}$ will be empty.

To enumerate all actually possible bit strings $(b_i)$, we can start with only the first candidate point $\{v_1\}$, then with the first two $\{v_1, v_2\}$ and so on, proceeding iteratively. As in Section 2\ref{sm-sec:DescriptionAlg}, we extend the notation $P_{(b_i)}$ to bit strings of length $|b|\leq |V_L|$ simply by setting
\begin{align*}
P_{(b_i)} = \bigcap_{i = 1}^{|b|} P^{b_i}(v_i) \, .
\end{align*}
For every candidate we add, we can go from a bit string of lenght $n$ to a bit string of lenght $n+1$ using the following recursion relation:
\begin{equation*}
\begin{aligned}
P_{(b_i)\oplus 1} &= P_{(b_i)} \cap P^1(v_{i+1})\\
P_{(b_i)\oplus 0} &= P_{(b_i)} \cap P^0(v_{i+1})\,.
\end{aligned}
\end{equation*}

This suggests a simple algorithm for the enumeration of all nonempty $P_{(b_i)}$, with string lenght $|b|=|V_L|$, and their associated bit strings.

The algorithm consists of $|V_L|+1$ steps, $k = 0, \dots, |V_L|$. At every step, we maintain a list $\ell_k$ of all bit strings $b$ of length $k$ for which $P_{(b_i)}$ is nonempty, together with the associated $P_{(b_i)}$. Initially, the list only consists of the empty bit string $\varnothing$. The associated set is $S_\varnothing = [0,1)$.

In the $k$-th step, we are given the list $\ell_{k-1}$ and want to compute $\ell_k$. So every bit string $b \in \ell_{k-1}$, we compute the two sets $S_{b0}$ and $S_{b1}$ as described above, which simply amounts to an intersection of intervals. In many cases, one of these sets will be empty (but not both). We then add to the list $\ell_k$ all those strings $b'=(b'_i)$ for which $P_{(b'_i)}$ is not empty. In the end, the list $\ell_{|V_L|}$ will contain all the bit strings $b$ for which $P_{(b_i)} \neq \varnothing$. See Figure~\ref{fig-split} for a visualization of this interval splitting procedure.

\begin{figure}[ht]
\centering \includegraphics{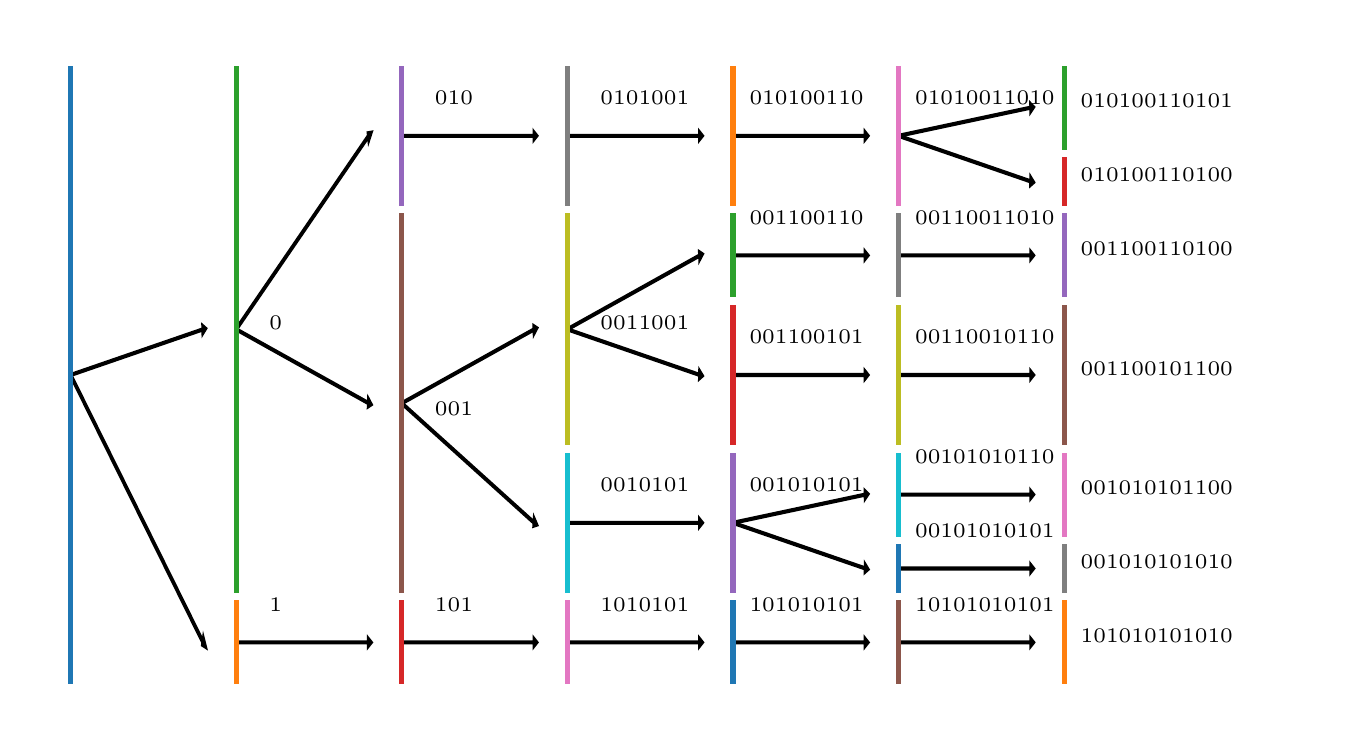}
\caption{Successive splitting of the intervals as more candidate points are added, for $L = 10$. The intervals correspond to the sets $P_{(b_i)}$, the strings to the partial bit strings $b$. \label{fig-split}. The large blue interval on the left hand side is the interval $[0,1)$ associated to the empty string. Steps where no interval is split are not shown.}
\end{figure}

\subsection{Gap bounds}

\noindent
With the method of enumerating local patches of the tiling described above, we will now describe how our lower bound and the upper bound of \cite{colbrook_2019} can be computed in the Fibonacci case.

We consider the Fibonacci Hamiltonian given by the standard Laplacian on $\ell^2(\mathbb{Z})$ with an electric potential given by the Fibonacci sequence. That is, if we denote by $f(n)$ the sequence with $f(n) = 1$ if the $n$-th character of the infinite Fibonacci sequence is $L$ and $f(n) = 0$ otherwise, we can define an operator on $\ell^2(\mathbb{Z})$ by
\begin{equation}
\label{FibonacciHamiltonian}
\begin{aligned}
H_{xy} = \begin{cases}
\alpha f(x) & x = y\\
-1 & |x - y| = 1\\
0 & \text{otherwise}
\end{cases}\,,
\end{aligned}
\end{equation}
where $\alpha \in \mathbb{R}$ is a constant. For $\alpha = 0$, the operator is the standard Laplacian on $\ell^2(\mathbb{Z})$ with absolutely continuous spectrum. As $\alpha$ moves away from $0$, gaps will open in the spectrum of the Laplacian. An enumeration of the local patch of the Fibonacci string can be made similarly to what is described above.

The upper bound of \cite{colbrook_2019} becomes quite simple in a one-dimensional setting like ours. Indeed, following the procedure described in \cite{colbrook_2019}, in order to compute the upper bound to the distance to the spectrum of a given energy $\lambda$, one is reduced to analyze the rectangular matrix  
\begin{align*}
B = \mathbf{1}_{[-(n+1), \dots, n+1]}(H - \lambda \mathbf{1}) \mathbf{1}_{[-n,n]}.
\end{align*}
Let $s$ be the lowest eigenvalue of $B^* B$, then the upper bound on the distance from $\lambda$ to the spectrum that is given in \cite{colbrook_2019} is simply $\sqrt{s}$. Although we could verify that the algorithm utilizing the Cholesky decomposition to compute $s$ given in \cite{colbrook_2019} is faster for large sparse matrices, a standard eigenvalue decomposition was sufficient to compute the values corresponding to the black line in Figure~\ref{fig-fib-both-bounds}.

To compute our lower bound, an enumeration of the possible tilings as outlined above was first made. Using $N = 6$ and $r = 1$, we split each interval $B_L(x)$ into an outer part $B_L(x) \backslash A(x)$, a middle part $A(x) \backslash B_l(x)$ and an inner part $B_l(x)$. It turned out that it was sufficient to always pick the maximal choice $A(x) = B_{L - 1}(x)$. This way, the outer part $B_L(x) \backslash A(x)$ has only two elements. Thus, to compute the norm in \eqref{eq-conditions-epsilon-recall}, we only need to compute
\begin{align*}
\mathbf{1}_{\overline{B_l}(x)} (H_{A(x)} - \lambda)^{-1} \mathbf{1}_{A(x)} H \mathbf{1}_{ B_L(x) \backslash A(x)}\mathbf{\tilde{x}}_i = \mathbf{1}_{\overline{B_l}(x)} (H_{A(x)} - \lambda)^{-1} \mathbf{1}_{A(x)} H \mathbf{\tilde{x}}_i  
\end{align*}
for the two vectors $\tilde{\bf x}_1 = \delta_{x+L}$ and $\tilde{\bf x}_2 = \delta_{-L+x}$. The constant $M^{-1}(x)$ may be computed as the lowest eigenvalue of the matrix
\begin{align*}
\mathbf{1}_{\overline{B}_{l(x)}} (H - \lambda) \mathbf{1}_{B_{L-1}(x)}\,.
\end{align*}
The values of this lower bound are plotted as the red curve in Figure~\ref{fig-fib-both-bounds}. It can be seen that our method can resolve even the fine fractal gap structure of the Fibonacci Hamiltonian and gives exact bounds on the extent of gaps of different orders of magnitude.

\section{4. Computational complexity of spectral computations}

\noindent
The extistence of the algorithm presented in this paper seemingly contradicts the statement in \cite{colbrook_2019} that it is impossible to compute spectra with error control in a general setting. In fact, even the spectrum of a diagonal operator on an infinite Hilbert space cannot be computed by an algorithm accessing the matrix elements one-by-one. It is only by requiring the additional structure of finite local complexity that such an algorithm can be found. This can be further elucidated in the framework of the \textit{solvability complexity index} (SCI) \cite{hansen2011, benartzi2015, colbrook2019foundations}: We show that computing the spectrum of an operator of finite local complexity is a problem with solvability complexity index $1$, whereas it is known that the index is $2$ in the general case \cite{colbrook2019foundations}.

\begin{definition} A \textit{computational problem} consists of a tuple $(\Omega, \Lambda, (\mathcal{M}, d), \Xi)$. Here $\Omega$ is the domain, or the set of problems. (In our case, $\Omega$ will be a set of operators on a Hilbert space.) The metric space $(\mathcal M, d)$ is the set of possible solutions, which in our case will be the power set of $\mathbb{R}$ equipped with the Hausdorff metric. The \textit{problem function} $\Xi : \Omega \rightarrow \mathcal M$ describes the exact solution of the problem (for example, the function that maps every operator to its spectrum). Finally, $\Lambda$ is a set of functions, $f_i : \Omega \to \mathbb{R}$, the \textit{evaluation functions}, which the algorithm uses to access information on the given object in $\Omega$.
\end{definition}

\begin{definition}
\label{def-sci}
A computational problem $(\Omega, \Lambda, (\mathcal M, d), \Xi)$ is said to have \textit{solvability complexity index $1$} if and only if there exists a sequence $\Gamma_n$ of functions $\Gamma_n : \Omega \to \mathcal M$ such that:
\begin{enumerate}
\item Every function $\Gamma_n$ can be computed by a finite number of divisions, radicals, and comparisons from a finite number of evaluations $f_i : \Omega \to \mathbb R$. More precisely, we can ask for the computations to be performed by a machine in the Blum-Shub-Smale (BSS) model \cite{blum1998}. Which and how many evaluations are performed can be decided based on the previous evaluations.
\item For every $A \in \Omega$, the computations $\Gamma_n(\Omega)$ converge to $\Xi$ in a controlled way, that is
\begin{align*}
d(\Gamma_n(A), \Xi(A)) < 2^{-n}
\end{align*}
\end{enumerate}
\end{definition}

By passing to a subsequence, this definition of solvability complexity index $1$ is equivalent to giving an algorithm $\Gamma_n$ and an error control function $E_n(A)$ such that $d(\Gamma_n(A), \Xi(A)) < E_n(A)$ for all $A \in \Omega$ and $E_n(A) \to 0$ as $n \to \infty$~\cite{benartzi2015}.

\begin{definition}
A \textit{Hamiltonian on a set} $\Gamma$ is a bounded self-adjoint operator on the Hilbert space $\ell^2(\Gamma, \mathcal H)$, for some finite-dimensional Hilbert space $\mathcal H$.
\end{definition}

\begin{definition}
A set $\Gamma \subseteq \mathbb{R}^d$ is called \textit{uniformly discrete} if there exists an $\epsilon > 0$ such that every two points $x, y \in \Gamma$ have distance $\|x-y\|_{\infty} \geq \epsilon$. The \textit{local patch around} $x$ \textit{at scale} $L$ is the set $\mathcal{C}_L(x):=\left\{ y-x | y \in \Gamma \cap B_L(x) \right\}$.  A set $\Gamma$ is said to have \textit{finite local complexity} if for every $L$, the set
\begin{align*}
\mathcal{C}_L = \big\{ \mathcal{C}_L(x) \,|\, x \in \Gamma \big\}\,,
\end{align*}
called the set of \textit{local patches}, is finite.
\end{definition}

\begin{definition}
\label{definition-finite-local-complexity}
A Hamiltonian $H$ on a set $\Gamma$ is said to have finite local complexity if there is a function $W : \mathbb{N} \to \mathbb{N}$ such that, for every $L \in \mathbb{N}$ and every $x \in \Gamma$, there is an $y \in B_{W(L)}(0) \cap \Gamma$ such that $H_{B_L(x)}$ and $H_{B_L(y)}$ are equivalent in the following sense:
\begin{align*}
H_{B_L(x)} = T_{y-x}^* U^* H_{B_L(y)} U T_{y-x}\,
\end{align*}
where $T_{y-x}$ is the translation operator and $U$ is a diagonal unitary operator. (The operator $U$ corresponds to a change of gauge for the magnetic Laplacian.) We say that $H$ has \textit{maximum hopping length} $m$ if the matrix entries $H_{xy} = 0$ for all $x, y \in \Gamma$ with $d(x,y) > m$
\end{definition}

\noindent
Now, the following no-go theorem is proven in \cite{colbrook_2019}.

\begin{theorem}
For every $N_{\max} \in \mathbb{R}$ and $m \in \mathbb R$, define the set $\Omega$ as the set of all Hamiltonians $H$ on uniformly discrete sets $\Gamma$ with $\norm{H}_{\text{op}} \leq N_{\max}$ and maximum hopping length $m$. Then the computational problem $(\Omega, \Lambda, (\mathcal M, d), \Xi|_{\Omega})$ has solvability complexity index $> 1$.
\end{theorem}

The proof of this no-go theorem is surprisingly simple and generalizes across all computational architectures. In fact, the authors remark that it even generalizes to the class of \textit{diagonal operators}. It is based on the fact that any algorithm can only see a finite part of a given infinite matrix, and by changing the matrix in one place outside this range, we can change the spectrum in a way that cannot be detected by the given algorithm. (The counterexample matrix thus has to be adapted to the algorithm.)

It was suggested in \cite[Remark 9.1]{benartzi2015} that it might be a fruitful subject of research to find some additional structure under which one can reduce the SCI, and that this could lead to new algorithms. We show here, using the algorithm described in the previous sections, that adding the structure of \textit{finite local complexity} lowers the SCI of computing the spectrum of self-adjoint operators with bounded dispersion from two to one.

\newcommand{\fc}{^{\mathrm{fc}}}
\begin{theorem}
\label{theorem-fc-complexity}
Let $\Omega\fc$ be the set of all tuples $(\Gamma, H)$, where $\Gamma$ is a uniformly discrete subset of $\mathbb{R}^d$ and $H$ is a Hamiltonian on the set $\Gamma$ with finite local complexity and finite hopping length. We need to introduce several evaluation functions to allow the algorithm to access the geometry of $\Gamma$. We assume that the set $\Gamma$ comes with an enumeration $(\gamma_k)_{k \in \mathbb{N}}$ of its points such that $\norm{\gamma_k}$ is nondecreasing.
\begin{enumerate}
\item for every $k_1, k_2 \in \mathbb{N}$, let $f^{\mathrm{mat}}_{k_1, k_2}(\Gamma, H) := \langle \gamma_{k_1} | H | \gamma_{k_2} \rangle$ evaluate the matrix entries of $H$.
\item for every $k \in \mathbb{N}$ and $l \in \{1, \dots, d\}$, let $f^{\mathrm{pos}}_{k, l}(\Gamma, H) := (\gamma_k)_l$ be the $l$-th coordinate of the $k$-th point.
\item for any $n \in \mathbb{N}$, let $f^{\mathrm{flc}}_n(\Gamma, H) := R$ be smallest radius $R \in \mathbb{R}$ such that every local patch of size $n$ (that is, every set $(\Gamma \cap B_n(x)) - x$ for $x \in \Gamma$) occurs for some $x \in B_R(0)$. (This function quantifies the finite local complexity)
\item let $f^{\mathrm{hop}}(H, \Gamma)$ be the maximum hopping length of $H$.
\end{enumerate}
77Let the set $\Lambda\fc$ comprise all the functions $f^{\mathrm{mat}}_{k_1, k_2}, f^{\mathrm{pos}}_{k, l}, f^{\mathrm{flc}}_n, f^{\mathrm{hop}}, $ and $f^{\mathrm{norm}}$. Finally, as before, define $(\mathcal{M}\fc, d)$ as the power set of $\mathbb{R}$ equipped with the Hausdorff distance and let $\Xi\fc(\Gamma, H) = \sigma(H)$, where $\sigma(H)$ denotes the spectrum of $H$.

The computational problem $(\Omega\fc, \Lambda\fc, (\mathcal M\fc, d), \Xi\fc)$ so defined is $\in \Delta_1^A$. (That is, its SCI is $1$.)
\end{theorem}

A detailed proof of this theorem will be given in a future publication. Basically, one can combine the upper bound on the distance to the spectrum from \cite[Supplementary Information, Theorem 3]{colbrook_2019} with the computable lower bound we provide here. The pertinent theorem on the upper bound from \cite{colbrook_2019} can be forumulated as follows:

\newcommand{\up}{^{\mathrm{up}}}
\newcommand{\low}{^{\mathrm{low}}}
\begin{theorem}
\label{theorem-colbrook-bound-sci}
Let $A \in \Omega_2$ and let $\Gamma_n\up(A), E\up(n)$ be computed using the algorithm in \cite{colbrook_2019} with size parameter $n$. Then $\Gamma_n\up(A) \to \sigma(A)$ and $E\up(n) \to 0$ as $n \to \infty$ and $\Gamma_n\up(n)$ is contained in the $E\up(n)$ neighbourhood of $\sigma(A)$. Moreover, $\Gamma_n$ can be implemented using finitely many arithmetic operations and comparisons on the matrix elements of $A$.
\end{theorem}

Using the algorithm in this article, we can also give the following statement on a computable lower bound for operators of finite local complexity:

\begin{theorem}
\label{theorem-our-bound-sci}
For the computational problem $(\Omega\fc, \Lambda\fc, (\mathcal M\fc, d), \Xi\fc)$, there exists an algorithm computing a spectral approximation $\Gamma\low_n(H)$ and an error estimate $E\low_n$ such that $\sigma(H)$ is contained in an $E\low_n$-neighbourhood of $\Gamma\low_n$.
\end{theorem}

From the combination of Theorems \ref{theorem-colbrook-bound-sci} and \ref{theorem-our-bound-sci}, we can clearly get an approximation of the spectrum that converges in Hausdorff distance with error control. (To get this set, simply compute the intersection of $\Gamma_n\low$ with the $E_n\up$ neighbourhood of $\Gamma_n\up$ and set $E_n = \max(E_n\up, E_n\low)$.) Hence the SCI of the problem $(\Omega\fc, \Lambda\fc, (\mathcal M\fc, d), \Xi\fc)$ is 1.

The proof of Theorem \ref{theorem-our-bound-sci} is based on the following results. 

\begin{theorem}
\label{theorem-always-bulk-gapped}
Let $H$ be an operator on a uniformly discrete set $\Gamma \subseteq \mathbb{R}^d$ with finite hopping length $m \in \mathbb{R}$. Then for every $\delta > 0$ there exists an $\epsilon > 0$ and an $L \in \mathbb{N}$, such that $H$ is $\epsilon$-bulk-gapped at scale $L$ for all energies $\lambda$ where the distance to the bulk spectrum $\rm{dist}(\lambda, \sigma(H)) > \delta$.
\end{theorem}

The tuple $(\epsilon, L)$ can be computed explicitly and  depends only on $\norm{H}$, on $m$ and on $\Gamma$. The proof of Theorem \ref{theorem-always-bulk-gapped}, which will be presented in a future publication, relies on a discrete version of the well-known Combes-Thomas estimates \cite{combes_thomas}.

\bibliography{quasicrystals}